\documentclass[10pt,onecolumn,draftcls,journal]{IEEEtran}
\usepackage[ansinew]{inputenc}
\usepackage{epsfig}
\usepackage[T1]{fontenc}
\usepackage{amsmath,enumerate,amssymb,hhline,verbatim}

\newtheorem{definition}{Definition}[section]
\newtheorem{thm}{Theorem}[section]
\newtheorem{proposition}[thm]{Proposition}
\newtheorem{lemma}[thm]{Lemma}
\newtheorem{corollary}[thm]{Corollary}
\newtheorem{exam}{Example}[section]
\newtheorem{result}{Result}
\newenvironment{proof}{\begin{IEEEproof}}{\end{IEEEproof}}

\newtheorem{remark}{Remark}[section]

\newcommand{\Q}{\mathbb{Q}}
\newcommand{\Z}{\mathbb{Z}}
\newcommand{\C}{\mathbb{C}}
\newcommand{\R}{\mathbb{R}}
\def\snr{\textnormal{SNR}}
\def\norm#1{\left\| #1 \right\|}
\def\diag{\textnormal{diag}}
\def\tr{\textnormal{Tr}}
\newcommand{\etal}{{\it et al. }}
\begin{document}

\title{An error event sensitive trade-off between rate and coding gain in MIMO MAC
\thanks{The research of T. Ernvall is supported in part  by the Academy of Finland grant \#131745.}
\thanks{The research of R. Vehkalahti is funded by the  Academy of Finland  grant  \#252457.}
\thanks{The research of H.-f. Lu was funded in part by Taiwan National Science Council under Grants NSC 100-2221-E-009-046 -MY3 and NSC 101-2923-E-009-001-MY3.}
\thanks{Part of this work appeared at ITW 2010 \cite{decaypaperi} and  ISIT 2012 \cite{EV}}
}

\author{Toni Ernvall, Jyrki Lahtonen \emph{Member, IEEE}, Hsiao-feng (Francis) Lu, \emph{Member, IEEE} and  Roope Vehkalahti, \emph{Member, IEEE}
\thanks{T. Ernvall is with the Turku Center of Computer Science, Turku, Finland and with with the Department of Mathematics and Statistics, FI-20014, University of Turku, Finland (e-mail:tmernv@utu.fi).}

 \thanks{ J. Lahtonen and R. Vehkalahti are with the Department of Mathematics and Statistics, FI-20014, University of Turku, Finland
(e-mails: \{ lahtonen, roiive\}@utu.fi).}
\thanks{H.-f. Lu is with the Department of Electrical Engineering, National Chiao Tung University, Hsinchu, Taiwan (e-mail:francis@mail.nctu.edu.tw). }
}

\maketitle

\newcommand{\La}{\mathbf{L}}
\newcommand{\h}{{\mathbf h}}

\newcommand{\D}{{\mathcal D}}
\newcommand{\F}{{\mathbf F}}
\newcommand{\HH}{{\mathbf H}}
\newcommand{\OO}{{\mathcal O}}
\newcommand{\G}{{\mathcal G}}
\newcommand{\A}{{\mathcal A}}
\newcommand{\B}{{\mathcal B}}
\newcommand{\I}{{\mathcal I}}
\newcommand{\E}{{\mathcal E}}
\newcommand{\PP}{{\mathcal P}}

\newcommand{\M}{{\mathcal M}}
\newcommand{\separ}{\,\vert\,}
\newcommand{\abs}[1]{\vert #1 \vert}

\begin{abstract}
This work considers space-time block coding for the Rayleigh fading multiple-input multiple-output (MIMO) multiple access channel (MAC).  If we suppose that the receiver is performing joint maximum-likelihood (ML) decoding, optimizing a MIMO MAC code against  a fixed error event leads to a situation where the  joint codewords of the users in error can be seen as a single user MIMO code. In  such a case   pair-wise error probability  (PEP) based determinant criterion of Tarokh et al.   can be used to upper bound the error probability.

It was already proven by Lahtonen  et al. that irrespective of the used  codes the determinants of the differences of codewords of the overall codematrices will decay as a function of the rates of the users.

This work will study this decay phenomenon further and derive  upper bounds for the decay of determinants corresponding any error event.
Lower bounds for the optimal decay are studied by constructions based on algebraic number theory and Diophantine approximation. For some error profiles the constructed codes will be  proven to be optimal.

While the perspective of the paper is that of  PEP, the final part of the paper proves how the achieved decay results  can be turned into statements about the diversity-multiplexing gain trade-off (DMT).

\end{abstract}

\section{Introduction}
Assume that we are to design a system for
 $U$ simultaneously transmitting synchronized users, each transmitting with
$n_t$ transmit antennas and, for simplicity so that we end up with
square matrices, over $Un_t$ channel uses. We can describe each
user’s signals as $n_t\times Un_t$ complex matrices. A multiuser
MIMO signal is then viewed as a $Un_t \times Un_t$
matrix obtained by using the signals of the individual users as blocks. So each
user is occupying $n_t$ rows in this overall transmission matrix.

We are not interested on the scenario, where single users have fixed finite codes, but of coding schemes, where each user has a family of growing codes.
A natural source for such schemes is to assume that each user has a lattice  and that the finite codes of individual users are carved
out of   these user specific lattices. We are then interested on the asymptotic behavior of error probabilities when we let the rates of each user grow.

In the Rayleigh fading multiple-input multiple-output (MIMO) multiple access channel (MAC)   it is a natural idea to develop a code design criteria by splitting the whole error probability space to separate error events based on which users are in error.  If we suppose that the receiver is performing joint maximum-likelihood (ML) decoding, optimizing a MIMO MAC code against  a fixed error event leads to a situation where the  joint codewords of the users in error can be seen as a single user MIMO code. In  such a case   pair-wise error probability based determinant criterion from \cite{TSC}    can be applied to upper bound the error probability.  This approach was  first taken in \cite{GB} and was developed further in \cite{CoGaBo} and  \cite{bb}.

The approach was  also  suggested for reaching  the MAC DMT \cite{Tse}    in \cite{CoGaBo}. Unfortunately  in \cite{LuHoVeLa} it was proved that the criteria given in \cite{CoGaBo}  is too tight,  at least in the case where the codes of the single users  are lattice  space-time codes. The  \emph{pigeon hole bound} in \cite{LuHoVeLa} proves that, irrespective of the code design, the determinants of the overall code matrices will decay with  a polynomial  speed with respect to the code-size  and methods, similar to those   used in \cite{EKPKL}, to prove the DMT optimality for  single user MIMO codes  do not work.

In this paper we will study  this decay phenomenon further and will  give general upper  and lower bounds for the decay of determinants corresponding any error event. The achieved results are then used to draw conclusions of the DMT of the analyzed codes.

The upper bounds we derive will be functions of rates of each users and are independent of the used  codes.
These results will reveal a trade-off between the rate and protection against error.
 It is possible to design a MAC code  in  such a way that the codewords corresponding each error event will have maximal possible rank. If we have such a code, then in each error event the code is protected by maximal possible diversity. However, in the spirit of  DMT, our bounds will reveal that when the rates of users of the corresponding error events are growing the determinants of the matrices offering the maximal diversity will decay. This interplay creates a trade-off between the coding gain and rate, which is close relative to DMT, but not the same concept.

The proofs of the  upper bounds are based on using pigeon hole principle, projection mappings  and properties of the determinant function.
The results are purely algebraic corresponding all lattice codes with certain properties. Only in the statement of the problem we are using Rayleigh fading MIMO MAC as a motivation.

While the first part of the paper provides upper bounds for the decay of the determinants in MIMO MAC,  the second part of the paper concentrates on giving code constructions where we have tried to optimize the codes from the decay point of view.
In  \cite{decaypaperi} it was proved that the two user single antenna code (BB-code)  given in \cite{bb} has in  some sense the best possible decay. In this paper we are now  giving a wast  generalization of  this code  to general MIMO MAC. The  codes we will build  fulfill the so called \emph{generalized rank criterion} and  do reach the pigeon hole bound. Our method will  share the basic structure with the original generalization of BB-code given in
\cite{BB}.

The single user codes we are using are based on the \emph{multi-block} codes from division algebras \cite{Lu} and \cite{YB07}.
This  approach  has been taken in several recent papers on MIMO MAC.  However, in these papers the full rank criterion has been achieved by using either transcendental  elements \cite{BB} (with exception of   $n_t=2$ case, which is dealt with algebraic elements having low degree) or algebraic elements with high degree \cite{matriisi}. Both of these methods make it extremely difficult to measure the decay of the codes and can  lead to bad decay. For a survey of these recent results, we refer the reader to thesis of Maya Badr \cite{B}.

Instead of the usual algebraic independence strategy  we will use \emph{valuation theory} to achieve the full rank condition. This  technical tool allows us  to use algebraic elements with low degree. By applying Galois theoretic method of Lu et al. \cite{LuHoVeLa} and  methods from Diophantine approximation, originally introduced by  Lahtonen  et al. in \cite{decaypaperi}, we will prove  that our codes achieve   good decay and in particular reach  the pigeon hole bound.

In Section \ref{DMTanalysis} we will show how the   lower bounds for the decay of our codes can be translated to lower bounds for the DMT. This analysis will reveal that in many cases the constructed codes do achieve the optimal diversity-multiplexing gain trade-off for low multiplexing gains.

\subsection{Multi-user codes, error events and corresponding  decay functions}\label{definitions}
In this section  we will show how the decay functions appear as a natural generalization of the minimum determinant criterion used in the design of  single user MIMO space-time block codes.

Let us suppose that we have  $U$ users, each having  $n_t$ antennas and that the receiver has $n_r$ antennas and  complete channels state information.
 We also suppose that the fading for each user  stays stable for $k$ time units, where $k\geq Un_t$. Let us refer to the channel matrix of the $i$th  user with $H_i\in M_{n_r\times n_t}(\C)$ and let us suppose that each of these have  i.i.d complex Gaussian random variables with zero mean and unit variance as coefficients. In this scenario the base station receives
$$
Y=\sum_{i=1}^{U} H_iX_i +W,
$$
where, $X_i\in M_{n_t\times k}(\C)$,  is the transmitted codeword from the $i$th user, and $W \in \M_{n_r\times k}(\C)$ presents the noise        having i.i.d complex Gaussian random variables as coefficients. In this scenario  multiuser MIMO signal is  a $Un_t \times k$ matrix where the rows $(j-1)n_t + 1, (j-1)n_t + 2, \dots, (j-1)n_t + n_t$ represent $j$th user's signal ($j=1, \dots, U$).

In order to keep the analysis in the paper streamlined and clean, we will assume that the  single user space-time codes are always of the following type.
\begin{definition}\label{latticedef}
A {\em matrix lattice} $\textbf{L} \subseteq M_{n_t\times k}(\C)$ has the form
$$
\textbf{L}=\Z B_1\oplus \Z B_2\oplus \cdots \oplus \Z B_r,
$$
where the matrices $B_1,\dots, B_r$ are linearly independent over $\R$, \emph{i.e.}, form a lattice basis, and $r$ is
called the \emph{rank}  or the \emph{dimension} of the lattice.
\end{definition}

Let us suppose that $x_i$ are matrices in $M_{n_t\times k}(\C)$.  Throughout the paper  we will  use  the notation
$$
\R(x_1,\dots, x_j)=\R x_1+\cdots+ \R x_j.
$$
 If $\textbf{L}$ is a lattice we also write
$$
\R(\textbf{L})=\R B_1\oplus \R B_2\oplus \cdots \oplus \R B_r.
$$

The finite single user codes used in the actual transmission are  of form
$$
\textbf{L} (M) = \left\{ \sum_{i=1}^{r} b_{i} B_{i} | b_{i} \in \Z, -M \leq b_{i} \leq M \right\},
$$
where $M$ is a given positive number.

In the single user MIMO transmission  the usual pairwise error probability based design criteria \cite{TSC}, leads to maximizing
$$
\min_{ X, Y\in \mathbf{L}(M), X\neq Y }\det((X-Y)(X-Y)^{\dagger}).
$$
As $X-Y\in  \mathbf{L}(2M)$, we can just as well try to maximize
\begin{equation}\label{mindet}
\min_{X\in \mathbf{L}(M), X\neq 0}\det(XX^{\dagger}).
\end{equation}
The first step  is of course that this value should always be non-zero. This is the \emph{rank criterion}.
Maximizing the value is called the \emph{minimum determinant criterion}. Let us now show how this criterion can be generalized to MIMO MAC context.

In the rest of the paper we   suppose that  each user applies a lattice space-time code $\textbf{L}_j \subseteq M_{n_t \times k}(\C), j=1, \ldots, U$. We  also assume that each user's lattice is of full rank $r=2n_tk$, and denote the basis of the  lattice $\textbf{L}_j$ by $B_{j,1}, \ldots, B_{j,r}$. Now the code associated with the $j$th user is a restriction of lattice $\textbf{L}_{j}$
$$
\textbf{L}_{j} (N_{j}) = \left\{ \sum_{i=1}^{r} b_{i} B_{j,i} | b_{i} \in \Z, -N_{j} \leq b_{i} \leq N_{j} \right\},
$$
where $N_{j}$ is a given positive number.

Using these definitions the \emph{$U$-user MIMO MAC code} is $(\textbf{L}_{1} (N_{1}), \textbf{L}_{2} (N_{2}), \dots, \textbf{L}_{U} (N_{U}))$. Let $\I_u=\{i_1,\dots,i_u\} \subseteq \{1,\dots, U\}$ be of size $u>0$.
We will then use the notation $M(X_{i_1},\dots, X_{i_u}) \in M_{un_t\times k}(\C)$, where we have stacked the codewords  $X_{i_j}$ from  the  users $\{i_1,\dots,i_u\}$ on top of  each other.

The following example shows how the different error events leads to different code design criteria.
\begin{exam}
Let us suppose that we have three user MIMO MAC channel. The channel equation can  now be written as
$$
[H_1,H_2, H_3 ]
\begin{pmatrix}
X_1\\
X_2\\
X_3
\end{pmatrix}
+W =H_1X_1 +H_2X_2 +H_3X_3 +W.
$$
Let us suppose that the receiver manages to decode the message $X_1$ of the first user. As we supposed that the receiver has  perfect channel state information we can simply subtract the matrix $H_1X_1$ from the channel equation.
Therefore, if we like to design the code against an error event where exactly  the second and the third user are in error, we can consider
simply  the code $(\mathbf{L_2},\mathbf{L_3})$  and try to maximize
\begin{equation}\label{mindet}
\min_{0\neq X_i\in \mathbf{L_i}(M)} \mathrm{det}(M(X_2,X_3) M(X_2, X_3)^{\dagger}) ,
\end{equation}
with the assumption $X_i\neq 0$. Here we again took a benefit of the fact that $\mathbf{L_2}$ and $\mathbf{L_3}$ are lattices and therefore additively closed. Similar analysis obviously provides us with a minimum determinant criterion for any error event.
\end{exam}

The previous example reveal that we should aim at maximizing the minimum determinants of subcodes corresponding  to any error events. The  design criterion for each error event is then the usual minimum determinant criterion, but with an extra assumption that the code matrices of the  users in error are non-zero.

\begin{definition}\label{generaldecay}
 Let $C_{U,n_t}=(\mathbf{L}_1,\mathbf{L}_2,\dots,\mathbf{L}_U)$ and $\I_u=\{i_1,\dots,i_u\}$ and write  $C_{U,n_t}^{\I_u}=(\mathbf{L}_{i_1},\mathbf{L}_{i_2},\dots,\mathbf{L}_{i_u})$. Define the decay function for $C_{U,n_t}^{\I_u}$ by setting
$$
D_{\I_u} (N_{i_1}, \dots, N_{i_u}) = \min_{X_{i_j} \in \mathbf{L}_{i_j}(N_{i_j})\setminus \{\mathbf{0}\}} \sqrt{\det(MM^{\dagger})},
$$
 where $M=M(X_{i_1},\dots,X_{i_u})$.
\end{definition}

For a special case $N_1=\dots=N_U=N$ we write
$$
D_{\I_u} (N) = D_{\I_u} (N_{i_1}=N, \dots, N_{i_u}=N).
$$

We also use notion
$$
D(N_{1}, \dots, N_{U}) = D_{\{1,\dots,U\}}(N_{1}, \dots, N_{U})
$$
and similarly as above
$$
D (N) = D_{\{1,\dots,U\}} (N_{1}=N, \dots, N_{U}=N).
$$

If we have a $U$-user code with decay function $D(N_{1}, \ldots, N_{U})$, such that $D(N_{1}, \ldots, N_{U}) \neq 0$ for all $N_{1}, \ldots, N_{U} \in \Z_{+}$,  we say that the code satisfies \emph{(generalized) rank criterion}.

 \section{Upper bounds for the decay of determinants}

There are several single user MIMO codes $\mathbf{L} \subseteq M_{n\times k}(\C)$  that satisfy
\[
\inf_{{\bf 0} \neq X \in \mathbf{L}} \det (XX^{\dagger}) > 0.
\]
Such a lattice code is said to have the \emph{non-vanishing determinant} (NVD) property.   This is obviously a guarantee
that the minimum determinants of all the finite codes $\mathbf{L}(M)$ are lower bounded and from the PEP point of view such a property is very desirable.

It is tempting to try to build such   MIMO MAC code $C_{U,n_t}$ that for any subset of user $\I_u=\{i_1,\dots,i_u\} \subset I$ the decay function would satisfy
$$
D_{\I_u} (N_{i_1}, \dots, N_{i_u})\geq \epsilon,
$$
for some fixed $\epsilon >0$, irrespective of the code sizes of the users $i_j$. Unfortunately this is not possible.  Already in \cite{LuHoVeLa} it was proved that no matter how we build the codes, the function $D(N)$ will decay as a function of $N$.  In this section we will  give upper bounds for  $D_{\I_u} (N_{i_1}, \dots, N_{i_u})$ for any subset $\I_u$ of user. The bounds will depend  on the size of the code of each user and will reveal a trade-off between rate and diversity.

In order to keep the paper easy  to understand   we delay the proof of Theorem \ref{maindecay} to the end of the paper. However, in order to give some idea of the proof, we introduce some concepts and  results that are most crucial and  demonstrate the basic ideas in Example \ref{mainexample}.

\subsection{The pigeon hole principle in a subspace}\label{pigeonsec}

Let us consider the space $M_{n\times k}(\C)$, where $k\geq n$ over the real numbers. It is  $2nk$-dimensional real vector space with a real inner product
$$
<X,Y>=\Re(Tr (XY^{\dagger})),
$$
 where $Tr$ is the matrix trace. This inner product also naturally defines a metric on the space $M_{n\times k}(\C)$, when
 setting $||X||= \sqrt{<X,X>}$.

\begin{definition}
Let us suppose that $A$ is a subspace of $M_{n\times k}(\C)$ . We then have that
$$
A^{\bot}:=\{x \,|\, x\in M_{n\times k}(\C),  <x,a>=0 \,\forall\, a\,\in A\},
$$
is an $\R$-linear subspace of $M_{n\times k}(\C)$.
\end{definition}

The following lemma is a collection of well known results from basic linear algebra.

\begin{lemma}\label{ortodecomp}
Let us suppose that $A$ is a $k$-dimensional subspace of $M_{n\times k}(\C)$. Then
$$
M_{n\times k}(\C)= A \oplus A^{\bot},
$$
 $ dim_{\R}(A^{\bot})= 2nk-dim_{\R}(A)=v$ and $A^{\bot}$ has  a $v$-dimensional orthonormal basis
$\{e_1,\dots, e_v\}$  with respect  to the real inner  product of $M_{n\times k}(\C)$.
\end{lemma}

Let us suppose that $x$ is an element of $M_{n\times k}(\C)$ and $A$ a subspace.
According to Lemma \ref{ortodecomp} we can now write $x$ uniquely in the form
\begin{equation}\label{decomposition}
x=a_1+a_2,
\end{equation}
where $a_1 \in A$ and $a_2 \in A^{\bot}$. This decomposition gives us a well defined  $\R$-linear mapping $\pi_{A^{\bot}}:M_{n\times k}(\C)\mapsto A^{\bot}$, where
\begin{equation}\label{decomposition2}
\pi_{A^{\bot}}(x)=a_2.
\end{equation}

\begin{lemma}[Pigeon hole principle in a subspace]\label{pigprincip}
Let us suppose that we have  a set of matrices $S=\{x_1, \dots, x_s\}\subset M_{n\times k}(\C)$, where
$||x_l||\leq M,\, \forall \,l$ and a   subspace  $A\subset M_{n\times k}(\C)$.
Then  the set $S$  includes matrices $x_i$ and $x_j$, where $x_i\neq x_j$ such that
$$
||\pi_{A^{\bot}}(x_i)-\pi_{A^{\bot}}(x_j)|| \leq 4h \frac{M}{s^{1/h}},
$$
where $h$ is the real dimension of the subspace $A^{\bot}$.
\end{lemma}

\begin{proof}
If the set $S$ includes such elements $x_i$ and $x_j$ that
$\pi_{A^{\bot}}(x_i)=\pi_{A^{\bot}}(x_j)$, we are done. We can now suppose that
$\pi_{A^{\bot}}(S)$ has $s$ different elements. The space $ A^{\bot}$ has an orthonormal basis
$E=\{e_1,\dots,e_h\} \subseteq M_{n\times k}(\C)$. We can now define a cube
$$
C(M) =\left\{ \sum_{l=1}^{h} b_{l} e_{l} \,|\, b_{l} \in \R, -M \leq b_{l} \leq M \right\},
$$
 which has volume $(2M)^k$. The projection is a shrinking map and therefore
$\pi_{A^{\bot}}(S)\subset C(M)$.

The cube $C(M)$ can be  divided into smaller cubes of side  length $4M/s^{1/h}$. The $C(M)$ gets now partitioned to
$(2M)^h/(4M/s^{1/h})^h=s/2^h$ small cubes. As $\pi_{A^{\bot}}(S)$ has $s$ elements,
there must be at least two that are in the same cube of side length $4M/s^{1/h}$.
\end{proof}

\begin{corollary}\label{latticepigprincip}
Let us suppose that we have an $l$-dimensional lattice $\mathbf{L} \subset M_{n\times k}(\C)$,  where each of the basis elements $y_i$ has norm less than a constant $K$, and a subspace  $A\subset M_{n\times k}(\C)$.
The subset $\mathbf{L}(M)$ has  such a  non zero matrix  $z$ that
$$
||\pi_{A^{\bot}}(z)|| \leq  2^{l/h+2}hK M^{\frac{h-l}{h}},
$$
where $h$ is the real dimension of the space  $A^{\bot}$ and $M\geq4$.
\end{corollary}

\begin{proof}
 Let us suppose that $y_1,\dots,y_l$ is the basis of the lattice $\mathbf{L}$. Every element of $\mathbf{L}$ is a $\Z$-linear combination of the basis elements and we have the inequality
$$ ||n_1y_1+\cdots+ n_l y_l || \leq \sum_{i=1}^l|n_i| ||y_i||. $$
Therefore, for any element  $x\in \mathbf{L}(M)$, we have
$$
||x|| \leq KM.
$$
  The set $\mathbf{L}(M/2)$ has  now at least $(M/2)^l$ elements. According to Lemma \ref{pigprincip} there is then a pair of elements  $x_i$ and $x_j$ such that
$$
 ||\pi_{A^{\bot}}(x_i)-\pi_{A^{\bot}}(x_j)||  \leq  4hK2^{l/h} \frac{M}{M^{l/h}}.
 $$

As $\pi_{A^{\bot}}$ is linear we have that  $||\pi_{A^{\bot}}(x_i- x_j)||  \leq   2^{l/h+2} hK\frac{M}{M^{l/h}}$.
Both $x_i$ and $x_j$ belong to  $\mathbf{L}(M/2)$ and therefore  $x_i-x_j \in \mathbf{L}(M)$ and we are done.
\end{proof}

\subsection{Upper bounds for the decay}

Let us give one more tool before stating the main result of the section.

Let us suppose that $X$ is a matrix in $M_{n\times k}(\C)$.
Let the set of row indices $J=\{1,\dots, n\}$ be partitioned to subset $J_1,\dots, J_r$. If we stack the rows in  $J_i$, we will get
a matrix in $M_{|J_i|\times k}(\C)$. Let us denote it with $X_i$. We then have the following.

\begin{lemma}\label{generalHadamard}
Let us suppose that $X\in M_{n\times k}(\C)$ and write $|J_i|=j_i$. We then have
$$
\det(XX^{\dagger})\leq \det(X_1X_1^{\dagger})\cdots \det(X_r X_r^{\dagger}) \leq \prod_{i=1}^{r} \frac{{||X_i||_F}^{2j_i}}{j_i^{j_i}}.
$$
\end{lemma}
\begin{proof}
The  first inequality is the generalized Hadamard inequality \cite{Gant} page 254.  The second inequality is  then an application of Hadamard inequality to the rows of $X_i$'s and AM-GM inequality.
\end{proof}

The following theorem states that if the single user codes are large the overall code corresponding to each error event will automatically include matrices with small determinants. This can be seen as a trade-off between rate and coding gain.

\begin{thm}[Generalized pigeon hole bound]\label{maindecay}
Let us suppose we have a MIMO MAC code $(\textbf{L}_{1} (N_{1}), \textbf{L}_{2} (N_{2}), \dots, \textbf{L}_{U} (N_{U}))$ of $U$ users. Let us also suppose  that each user has $n_t$ transmission antennas and that the individual codes $\textbf{L}_i$ are   $2kn_t$-dimensional  lattice codes in $M_{n_t\times  k}(\C)$, where $k \geq Un_t$. We then have that
$$
D_{\I_u}(N_{i_1}, \dots, N_{i_u}) \leq K \prod_{l=1}^{u-1} N_{i_l}^{-\frac{n_{t}^2 (u-l)}{k-n_{t}(u-l)}}
$$
and in particular
$$
D_{\I_u}(N) \leq \frac{K}{N^{\alpha}},
$$
where $\alpha=\sum_{l=1}^{u-1} \frac{n_{t}^2 (u-l)}{k-n_{t}(u-l)}$ and $K$ is a fixed constant.
\end{thm}

The proof of the previous theorem is given only in the Appendix, but let us describe the proof in a simple case of three users in the following example.
\begin{exam}\label{mainexample}
 Let us suppose we have $U=3$ users, each transmitting with $n_t=1$ antenna, and that the code length is $k=Un_t=3$. For simplicity, let us also assume  that $N_1=N_2=N_3=N$. Let us now study  the behavior of the decay function $D(N)$ in this scenario.

Let us first fix some small $C_3 \in \mathbf{L}_{3}(N)$.  Hence $||C_3||=\OO(1)$.  Let  $W_3 =  \R(C_3) $ and let $V_3 = W_{3}^{\bot}$ be its orthogonal complement. The corresponding orthogonal projection is $\pi_{3}: M_{1 \times 3}(\C) \rightarrow V_3$.

A subspace $V_3$ has $\dim_{\R}(V_3)=2n_{t}k-\dim_{\R}(W_3)=6-2=4$ so the image $\pi_{3}(\mathbf{L}_{2}(N))$ falls into a $4$-dimensional hypercube with side length of size $\OO(N)$. We also have $|\mathbf{L}_{2}(N)|=\theta(N^{6})$ and therefore by pigeon hole principle, we have such $C_{2} \in \mathbf{L}_{2}(N))$ that
$$
\pi_{3}(C_{2})=\OO(\sqrt[4]{\frac{N^{4}}{N^{6}}})=\OO(N^{-\frac{1}{2}}).
$$

Now similarly build $V_{2}=W_{2}^{\bot}$ by setting $W_{2} = \R(C_2 , C_3)$. This gives $\dim_{\R}(V_{2})=2$. And again we find such $C_{1}$ that
$$
\pi_{2}(C_{1})=\OO(\sqrt{\frac{N^{2}}{N^{6}}})=\OO(N^{-2})
$$
with $\pi_{2}: M_{1 \times 3}(\C) \rightarrow V_{2}$ being an orthogonal projection.

Hence we have a matrix $A$ in our code such that the determinant $\det(AA^{\dag})$ is by Lemmas \ref{matriisitulo} and \ref{generalHadamard} of size
$$
\OO(N^{2 \cdot (-\frac{1}{2})} \cdot N^{2 \cdot (-2)})) = \OO(N^{-5})
$$
and hence $D(N)=\OO(N^{-\frac{5}{2}})$.

Note that in the special case of this example $A$ is a square matrix and hence it is not necessary to use Lemma \ref{matriisitulo}. It is enough to notice that $\det(AA^{\dag}) = |\det(A)|^2$ and
$$
\det(A)
=\begin{pmatrix}
C_1\\
C_2\\
C_3
\end{pmatrix}
=\begin{pmatrix}
\pi_{2}(C_1)\\
C_2\\
C_3
\end{pmatrix}
=\begin{pmatrix}
\pi_{2}(C_1)\\
\pi_{3}(C_2)\\
C_3
\end{pmatrix}
$$
and then use Lemma \ref{generalHadamard} to estimate the size of the determinant $\det(AA^{\dag})$.
\end{exam}

\section{Lower bounds for the decay by explicit constructions}\label{construdiscussion}
Due to the strongly algebraic nature  of our constructions in Section \ref{construction},  we will only preview the results of that section here and give a simple example of our approach.

 The main result in \ref{construction} is   a completely general construction for MIMO MAC codes where each of the $U$ single users with $n_t$ antennas, will have a lattice code $\mathbf{L}_i$ in $M_{n_t\times Un_t}(\C)$ and where the overall MAC code $C_{U,n_t}=(\mathbf{L}_1,\mathbf{L}_2,\dots,\mathbf{L}_U)$ has  the generalized full rank property and promising decay properties.  The goal of these constructions is to give a general construction of MIMO MAC codes that would have as good as possible decay in every error event.

The natural comparison for the following result is the upper bound of Theorem \ref{maindecay}.
Let $C_{U,n_t}$  be a MAC code build in Section \ref{construction}. We then have the following two  results.
\begin{result}[Theorem \ref{inertdecay3}]
For a code $C_{U,n_t}$ there exists a constant $K>0$ such that
$$
D_{\I_u} (N_{i_1}, \dots, N_{i_u}) \geq \left\{
  \begin{array}{l l}
    \frac{K}{(N_{i_1}\cdots N_{i_u})^{(U-1)n_t}} & \quad \text{if $u>1$}\\
    K & \quad \text{if $u=1$}.
  \end{array} \right.
$$
In particular we have
$$
D_{\I_u} (N) \geq \left\{
  \begin{array}{l l}
    \frac{K}{N^{u(U-1)n_t}} & \quad \text{if $u>1$}\\
    K & \quad \text{if $u=1$}.
  \end{array} \right.
$$
\end{result}

The following result proves that in the case where we let only the rate of a fixed single user grow  we achieve the optimal decay.
\begin{result}[Corollary \ref{pigeon2}]
For a code $C_{U,n_t}$ there exists constants $k>0$ and $K>0$ such that
$$
\frac{k}{N^{(U-1)n_t}} \leq D(N_1=N, N_2= \ldots =N_U=1) \leq \frac{K}{N^{(U-1)n_t}}.
$$
\end{result}

Let us now give an example of our constructions.
\begin{exam}
Let us  suppose that we have $U$ users, each having a single antenna.

We can now find a degree  $U$ cyclic extension $L/\Q(i)$ and  have the Minkowski embedding $\psi: L\mapsto \C^U$, where
$$
\psi(x)=(x,\sigma(x),\dots, \sigma^{n-1}(x)),
$$
for  $x\in L$. If we concentrate on  the ring of algebraic integers $\OO_L$, we have that
$\psi(\OO_L)$ is a $2n$-dimensional lattice in $\OO_L$.

Let us  now suppose that $p$ is a prime number which is totally inert  in the extension $L/\Q$.
We can then modify the embedding $\psi$ to get $U$ single user lattice codes $\psi_{i, p^{-1}}(\OO_L)$, where
$$
\psi_{i, p^{-1}}(x)=(x,\sigma(x), \dots p^{-1}\sigma^{i-1}(x) \dots, \sigma^{U-1}(x)).
$$

A  $U$ user MAC code $C_{U,1}$ can now be defined by
$$
(\psi_{1,p^{-1}}(\OO_L),\psi_{2,p^{-1}}(\OO_L) \dots,  \psi_{U,p^{-1}}(\OO_L)),
$$
where each  of  the single user codes are previously defined  $2U$-dimensional lattices in $\C^U$.

The overall code matrices in this MAC code then have the form
\begin{equation}\label{singexam}
\left(
         \begin{array}{ccccc}
           p^{-1}x_1   &         \sigma(x_1) &       \sigma^2(x_1) & \ldots &         \sigma^{U-1}(x_1) \\
                 x_2   & p^{-1}  \sigma(x_2) &       \sigma^2(x_2) & \ldots &         \sigma^{U-1}(x_2) \\
                 x_3   &         \sigma(x_3) & p^{-1}\sigma^2(x_3) & \ldots &         \sigma^{U-1}(x_3) \\
           \vdots      & \vdots              & \vdots              &        & \vdots                    \\
                 x_U   &         \sigma(x_U) &       \sigma^2(x_U) & \ldots & p^{-1}  \sigma^{U-1}(x_U) \\
         \end{array}
       \right).
\end{equation}
The key  to the generalized rank criteria is the choice of the element $p$.
 By analyzing the $p$-adic valuation of the  determinant of the overall code matrix \eqref{singexam}, we find that when all the rows are non-zero,  the valuation of the determinant can not be $\infty$ and therefore  the determinant can not be zero.

 We could have  also used  for example transcendental element on place of $p^{-1}$, but as we will later see, from the decay point of view it is crucial  that the element $p^{-1}$ is from the field $L$. By such choice of the diagonal  element, all the elements in codematrices are from a low degree number field $L$ and we can effectively use results from Diophantine approximation to prove results about the decay.

If we suppose that $U=3$ Theorem \ref{inertdecay3} gives us that
$$
D_{\{1,2,3\}} (N,N,N) \geq \frac{K}{N^{6}},
$$
for some constant $K$.
\end{exam}

\section{ Lower Bounds for the DMT of our constructions}\label{DMTanalysis}
It was proved  in \cite{EKPKL} and \cite{TV}  that if a  $2n_tk$ dimensional lattice code $L\subset M_{n_t\times k}(\C)$ has the NVD property, then it is DMT optimal in the single user $n_t\times n_r$ MIMO channel.  It is  a natural idea to see what can be said about the DMT of the codes constructed in this paper based on the lower bounds for the decay of determinants.

Having obtained in Theorem 8.3 a lower bound on the minimum determinant among all the nonzero code matrices in code $C_{U,n_t}$, in this section we will apply this  bound to investigate the diversity-multiplexing gain tradeoff (DMT) achieved by $C_{U,n_t}$.
We note that in this section the analyzed codes $C_{U,n_t}$ are the MAC codes we constructed in Section \ref{construction}.

\subsection{Some Preliminaries}

We first introduce a power constraint on the transmitted signal matrices and reformulate the MIMO multiple-access channel as
\[
Y \ = \ \sum_{i=1}^U\kappa_i  H_i  X_i + W,
\]
where $H_i \in M_{n_r \times n_t}(\C)$ is the channel matrix of the $i$th user, and $W\in M_{n_r \times k}(\C)$ is the white noise matrix; both are defined as before. Realizations of $H_i$ are known perfectly to the receiver but are unknown to the users. $X_i \in {\bf L}_i \left( N_i \right) \subset M_{n_t \times k}(\C)$ is the signal matrix transmitted by the $i$th user. $\kappa_i \in \R^+$ is an amplification factor such that the average signal-to-noise power ratio (SNR) of the $i$th user equals $\snr$, \emph{i.e.},
\[
\kappa_i^2 \frac{1}{\abs{{\bf L}_i \left( N_i \right)}} \sum_{X_i \in {\bf L}_i \left( N_i \right)} \norm{X_i}^2 \ = \ \snr.
\]
Having specified $N_i$, the transmission rate, in bits per channel use, of the $i$th user is
\[
R_i  \ = \ \frac{1}{k} \log_2 \abs{{\bf L}_i \left( N_i \right)}.
\]
By increasing the rate $R_i$ as a linear function of $\log_2 \snr$ as $\snr \to \infty$, following \cite{ZheTse}, we say the $i$th user transmits at multiplexing gain $r_i$ if
\[
\lim_{\snr\to\infty} \frac{R_i}{\log_2 \snr} \ = \ r_i.
\]
Equivalently, we shall adopt the dotted-equality notation~\footnote{Let $f(\snr)$ and $g(\snr)$ be two functions of $\snr$. We say $f(\snr) \doteq g(\snr)$ if $\lim_{\snr\to\infty} \frac{f(\snr)}{\log_2 \snr} = \lim_{\snr\to\infty} \frac{g(\snr)}{\log_2 \snr}$. The dotted inequalities such as $\dot\geq$, $\dot\leq$, $\dot>$, and $\dot<$ are defined similarly.} introduced in \cite{ZheTse} to rewrite the above as
\[
\abs{{\bf L}_i \left( N_i \right)} \ \doteq \ \snr^{k r_i}.
\]
As ${\bf L}_i \left( N_i \right)=\left\{ \sum_{j=1}^{2 n_t k} b_j B_{i,j} | b_j\in \Z, \abs{b_j} \leq N_i\right\}$ and $\abs{{\bf L}_i \left( N_i \right)} = \left( 2 N_i+1\right)^{2 n_t k}$, transmitting at multiplexing gain $r_i$ implies that
\begin{equation}
N_i \ \doteq \ \snr^{\frac{r_i}{2 n_t}}.
\end{equation}
On the other hand, note that the basis matrices $B_{i,1}, \ldots, B_{i,2 n_t k}$ are constant matrices and are independent of $\snr$. It can be shown that for $N_i \dot> \snr^0$
\[
\frac{1}{\abs{{\bf L}_i \left( N_i \right)} } \sum_{X_i \in {\bf L}_i \left( N_i \right)} \norm{X_i}^2 \ \doteq \ N_i^2.
\]
Hence, the amplification factor associated with a fixed multiplexing gain $r_i$ is
\begin{equation}
\kappa_i^2 \ \doteq \ \snr^{1-\frac{r_i}{n_t}}.
\end{equation}
Finally, we say the code $C_{U,n_t}$ achieves diversity gain $d(r_1, \ldots, r_U)$ if
the codeword error probability $P_{\text{cwe}}(r_1, \ldots, r_U)$ subject to the joint maximal-likelihood decoding of $(X_1, \ldots, X_U)$ at the receiver satisfies
\[
P_{\text{cwe}}(r_1, \ldots, r_U) \ \doteq \ \snr^{-d(r_1, \ldots, r_U)}.
\]
The function $d(r_1, \ldots, r_U)$ is also termed the MIMO MAC DMT for code $C_{U,n_t}$. It is known \cite{Tse,LuHoVeLa} that $d(r_1, \ldots, r_U)$ is upper bounded by
\begin{multline}
\lefteqn{d(r_1, \ldots, r_U)}
 \leq  \min\left\{ d^*_{u n_t, n_r} \left( \sum_{i \in {\cal I}} r_i \right) | \begin{array}{l}
{\cal I} \subseteq \{1,2,\ldots,U\}, \\
\abs{{\cal I}}=u, u=1,\ldots,U
\end{array} \right\} \yesnumber \label{eq:upperbound1}
\end{multline}
where the RHS represents the best possible diversity gain that can be achieved by any MIMO MAC codes when transmitted at multiplexing gains $r_1, \ldots, r_U$, respectively. The function $d^*_{m,n}(r)$ is the optimal DMT for a point-to-point MIMO channel with $m$ transmitting and $n$ receiving antennas and transmitting at multiplexing gain $r$. In particular, $d^*_{m,n}(r)$ is a piecewise linear function obtained by joining the points $(r,(m-r)(n-r))$ for $r=0,1,\ldots,\min\{m,n\}$.

\subsection{Lower  Bounds on $d(r_1, \ldots, r_U)$}

To analyze $d(r_1, \ldots, r_U)$ for code $C_{U,n_t}$, let ${\cal E}_{\cal I}$ denote the event that only the signals of users in set ${\cal I}$, ${\cal I} \subseteq {\cal U} = \{1,2,\ldots,U\}$, are erroneously decoded. Clearly, the overall error event is ${\cal E}=\bigcup_{{\cal I} \subseteq {\cal U}} {\cal E}_{\cal I}$, and the codeword error probability is upper bounded by
\begin{equation}
P_{\text{cwe}}(r_1, \ldots, r_U) \ = \ \Pr \{ {\cal E} \} \ \leq \ \sum_{{\cal I} \subseteq {\cal U}} \Pr\left\{ {\cal E}_{\cal I} \right\}. \label{eq:pcwe}
\end{equation}
To further upper-bound the probability of error event ${\cal E}_{\cal I}$ we employ the bounded-distance decoder introduced in \cite{EKPKL}. Specifically, given the channel matrices $H_i$, $i=1,\ldots,U$, the bounded-distance decoder searches for code matrices $(X_1, \ldots, X_U)$ such that $\norm{Y - \sum_{i=1}^U \kappa_i H_i X_i} < \frac12 d_{\min}(H_1, \ldots, H_U)$, where $d_{\min}(H_1, \ldots, H_U)$ is the minimum distance among matrices $\sum_{i=1}^U \kappa_i H_i X_i$ for all $X_i \in {\bf L}_i(N_i)$, that is,
\begin{multline*}
d_{\min}(H_1, \ldots, H_U)
= \min\left\{ \norm{\sum_{i=1}^U \kappa_i H_i \Delta X_i} |
\begin{array}{l}
\Delta X_i \in {\bf L}_i(2N_i) \text{ and } \\
 \text{ not all } \Delta X_i = {\bf 0}
 \end{array} \right\}.
\end{multline*}
It should be noted that as the matrices $H_1, \ldots, H_U$ are random, the minimum distance $d_{\min}(H_1, \ldots, H_U)$ is indeed a nonnegative random variable. Furthermore, it is clear that
\[
\Pr \{ {\cal E} \} \ \leq \ \Pr \left\{ \norm{W} \ \geq \ \frac12 d_{\min}(H_1, \ldots, H_U)  \right\}
\]
where $W$ is the noise matrix, and the RHS represents an upper bound on the probability of decoding error/failure of such bounded distance decoder. Now focusing on the error event ${\cal E}_{\cal I}$, where only the signals of users in set ${\cal I}=\{i_1, \ldots, i_u\}$ and $\abs{{\cal I}}=u$ are decoded in error. We define the corresponding minimum distance in this case by
\begin{multline*}
d_{\min}(H_{i_1}, \ldots, H_{i_u})
= \min\left\{ \norm{\sum_{i \in {\cal I}} \kappa_i H_i \Delta X_i} |
\begin{array}{l}
\Delta X_i \in {\bf L}_i(2N_i) \text{ and }\\
\text{ not all } \Delta X_i = {\bf 0}
\end{array}\right\}.
\end{multline*}
Similarly, it can be shown that
\begin{equation}
\Pr \{ {\cal E}_{\cal I} \} \ \leq \ \Pr \left\{ \norm{W} \ \geq \ \frac12 d_{\min}(H_{i_1}, \ldots, H_{i_u})  \right\}. \label{eq:bdd}
\end{equation}
The lower bound on the minimal determinant given in Theorem 8.3 then allows us to obtain a lower bound on $d_{\min}(H_{i_1}, \ldots, H_{i_u}) $ and therefore leads to a further upper bound on the error probability $\Pr \{ {\cal E}_{\cal I} \}$. The proof of the following Theorem will be given in Appendix.
\begin{thm} \label{thm:E_I}
For a MIMO MAC code $C_{U,n_t} \subset M_{U n_t \times k}(\C)$ of $U$ users defined as before, assume the users transmit at multiplexing gains $r_1, \ldots, r_U$, respectively. The probability of event ${\cal E}_{\cal I}$ that only the signals of users in set ${\cal I}$, ${\cal I} \subseteq {\cal U}$, are erroneously decoded is upper bounded by
\begin{equation}
 \Pr \{ {\cal E}_{\cal I} \} \ \dot\leq \ \left\{
 \begin{array}{ll}
 \snr^{-d^*_{un_t, n_r}(U\sum_{i \in {\cal I}} r_i)}, & \text{ if $\abs{{\cal I}}=u > 1$,}\\
 \snr^{-d^*_{n_t, n_r}(r_i)}, & \text{ if ${\cal I}=\{i\}$.}
 \end{array} \right.
\end{equation}
\end{thm}

Applying Theorem \ref{thm:E_I} to the union bound of $\Pr\{{\cal E}\}$ in \eqref{eq:pcwe} we immediately arrive at the following corollary. This gives a lower bound on the MAC DMT for code $C_{U,n_t}$.

\begin{corollary}[Lower bound on $d(r_1, \ldots, r_U)$]
For a MIMO MAC code $C_{U,n_t} \subset M_{U n_t \times k}(\C)$ of $U$ users defined as before, the corresponding MAC DMT is lower bounded by
\begin{multline}
d(r_1, \ldots, r_U)   
\hspace{0.1in} \geq \min\left\{
\begin{array}{l}
d_{u n_t, n_r}^* \left( U \sum_{i \in {\cal I}} r_i \right), d_{n_t, n_r}^*(r_j) |  \\
\qquad {\cal I} \subseteq {\cal U}, \abs{{\cal I}}=u > 1, j \in {\cal U}
\end{array} \right\}.
\label{eq:gendmtlb}
\end{multline}
In particular, for the symmetric MIMO MAC where $r_1=r_2=\cdots=r_U=r$, we have
\begin{multline}
\hspace{-0.15in}d(r, \ldots, r)  \geq   \min\left\{ d_{u n_t, n_r}^* \left( U u r \right), d_{n_t, n_r}^*(r) | u=2,\ldots,U \right\}\\
= \left\{
\begin{array}{ll}
\min\{ d_{2 n_t, n_r}^* \left( 2 U r \right), d_{n_t, n_r}^*(r)\}, & \text{ if $r\in[0,\theta]$}\\
\min\{ d_{U n_t, n_r}^* \left( U^2 r \right), d_{n_t, n_r}^*(r)\} & \text{ if $r\in[\theta, \frac{n_t}{U} ]$}
\end{array}
\right.,
 \label{eq:symdmtlb}
 \end{multline}
 where $\theta=\min\{ \frac{n_t}{U}, \frac{n_r}{U(U+2)}\}$.
\end{corollary}
\begin{proof}
It simply follows from \eqref{eq:pcwe}, after noting  that when  $\snr \to \infty$ we have
\[
P_{\text{cwe}}(r_1, \ldots, r_U) \ \leq \ \sum_{{\cal I} \subseteq {\cal U}} \Pr\left\{ {\cal E}_{\cal I} \right\} \ \doteq \ \max_{{\cal I} \subseteq {\cal U}} \Pr\left\{ {\cal E}_{\cal I} \right\}.
\]
The second equality in \eqref{eq:symdmtlb} can be shown by arguing similarly as in \cite[Sec. VIII]{Tse} that
\begin{multline*}
\min\left\{ d_{u n_t, n_r}^* \left( u \gamma \right)| u=2,\ldots,U \right\}
= \left\{
\begin{array}{ll}
d_{2 n_t, n_r}^* \left( 2 \gamma \right) & \text{ if $\gamma\in[0,\theta']$}\\
d_{U n_t, n_r}^* \left( U \gamma \right) & \text{ if $\gamma\in[\theta',n_t ]$}
\end{array}
\right.
\end{multline*}
where $\theta' = \min\{n_t, \frac{n_r}{U+2}\}$ and by setting $\gamma=U r$.
\end{proof}

To summarize, in this section we have presented a general lower bound \eqref{eq:gendmtlb} on the DMT performance of code $C_{U,n_t}$. The lower bound for the symmetric case is given in \eqref{eq:symdmtlb}. While the  optimal MAC DMT \cite{Tse} for all possible MIMO MAC codes is given by
\begin{equation}
\begin{array}{l}
d^*_{U,n_t,n_r}(r_1, \ldots, r_U) =
\hspace{0.4in} \min\left\{
\begin{array}{l}
d_{u n_t, n_r}^* \left(  \sum_{i \in {\cal I}} r_i \right), d_{n_t, n_r}^*(r_j) |  \\
\qquad {\cal I} \subseteq {\cal U}, \abs{{\cal I}}=u > 1, j \in {\cal U}
\end{array} \right\}.
\end{array}\label{eq:optimaldmt}.
\end{equation}

Below we apply the bounds to study the DMT performance of code $C_{U,n_t}$ over some MIMO MACs. For simplicity, we will focus only on the symmetric case. In Fig. \ref{fig:1} we present the bounds on the DMT performance of $C_{U,n_t}$ for $n_t=2$, $n_r=4$ and $U=3$. The ``optimal'' DMT curve represents the optimal MIMO MAC DMT $d^*_{U,n_t,n_r}(r_1, \ldots, r_U)$ given in \eqref{eq:optimaldmt}. It is known \cite{Tse} that the maximal possible multiplexing gain for this channel is $\frac{\min\{U n_t,n_r\}}{U}=\frac43$. This means that whenever $r > \frac43$ the corresponding diversity gain must be zero, and communications over this channel cannot be reliable. Furthermore, Tse \etal   \cite{Tse} show that for $r \leq \frac{\min\{U n_t,n_r\}}{U+1}=1$ the optimal DMT is dominated by the single-user performance. The remaining region where $1 \leq r \leq \frac43$ is termed the ``antenna pooling'' region \cite{Tse} and the DMT performance is dominated by the case when all user's signals are erroneously decoded. From Fig. \ref{fig:1} we see that the code $C_{u,n_t}$ is in fact MAC-DMT optimal for $r \leq 0.24$.

In Fig. \ref{fig:2} we present the bounds on the DMT performance of code $C_{U,n_t}$ for $n_t=2$, $n_r=8$ and $U=3$, corresponding to the case without antenna pooling region. The optimal MAC-DMT is completely dominated by the single-user performance. It is seen from Fig. \ref{fig:2} that the code $C_{U,n_t}$ remains to be DMT optimal whenever $r \leq \frac{23}{45} \approx 0.311$. Fig. \ref{fig:3} shows the bounds on the DMT performance of code $C_{U,n_t}$ for $n_t=3$, $n_r=6$ and $U=2$. We see that the code $C_{U,n_t}$  is DMT optimal whenever $r \leq 0.6$.

\begin{figure}[h!]
\[
\includegraphics[width=0.9\columnwidth]{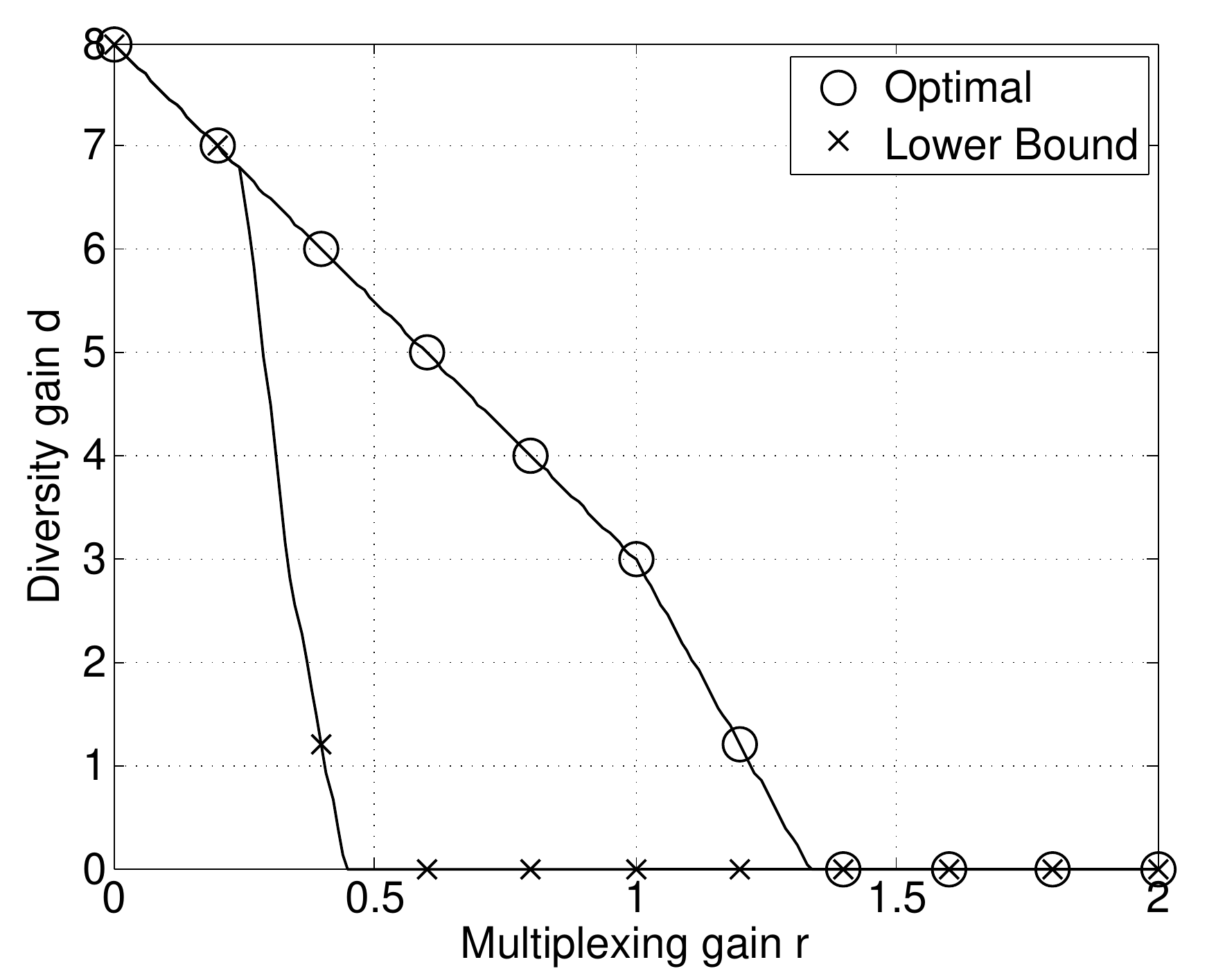}
\]
\caption{Bounds on the DMT performance of code $C_{U,n_t}$ for $n_t=2$, $n_r=4$ and $U=3$} \label{fig:1}
\end{figure}

\begin{figure}[h!]
\[
\includegraphics[width=0.9\columnwidth]{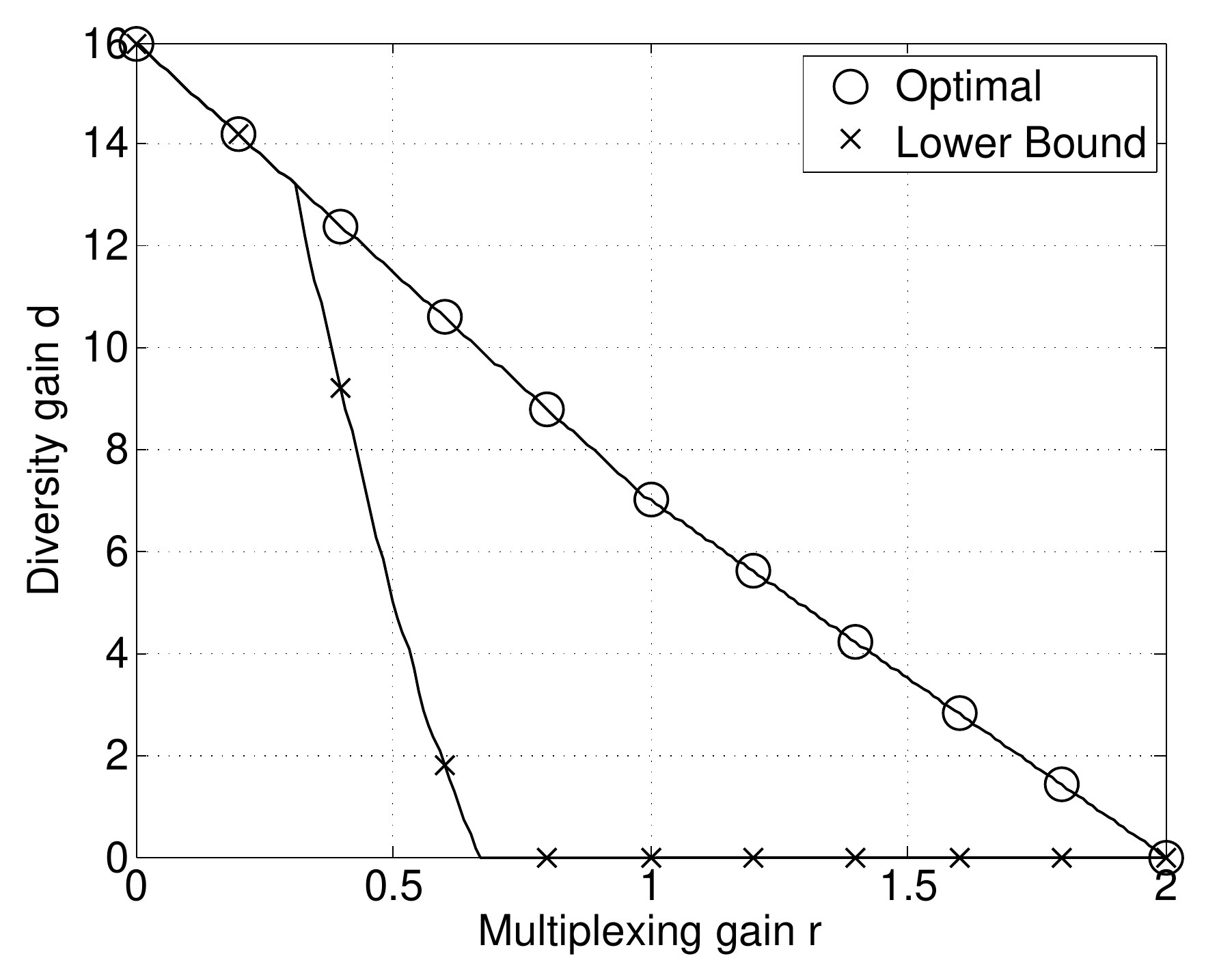}
\]
\caption{Bounds on the DMT performance of code $C_{U,n_t}$ for $n_t=2$, $n_r=8$ and $U=3$} \label{fig:2}
\end{figure}

\begin{figure}[h!]
\[
\includegraphics[width=0.9\columnwidth]{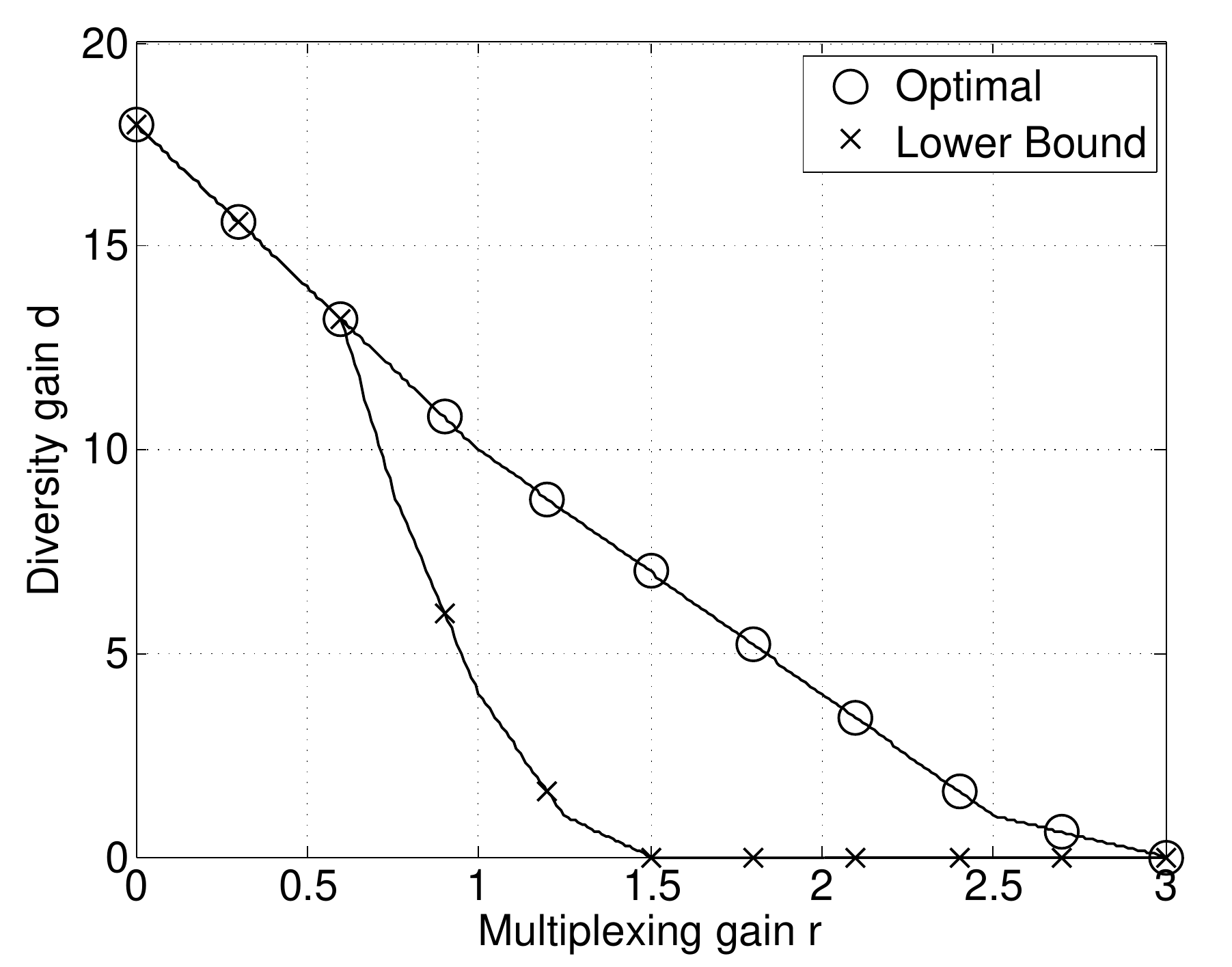}
\]
\caption{Bounds on the DMT performance of code $C_{U,n_t}$ for $n_t=3$, $n_r=6$ and $U=2$} \label{fig:3}
\end{figure}

\section{Constructions}\label{construction}
In this section we are giving our general constructions, having properties discussed in Section \ref{construdiscussion}. As explained there
the main idea is to build codes having as good as possible decay. Unfortunately we are forced to use several techniques from different parts of number theory and will assume that the reader has working knowledge on the topic. We will also skip some of the proofs as they are purely number theoretical and rather standard.

First in Subsection \ref{2user} we analyze the single antenna two user case thoroughly. Then in Subsection \ref{construction} we will generalize this construction for $U$ users with $n_t$ transmit antennas.

\subsection{A 2-user code}\label{2user}
In \cite{bb}  Badr and Belfiore introduced  a $2$-user single antenna MAC  code where the matrix coefficients were from the field $\Q(i,\sqrt{5})$ and had the generalized full rank property. It was proved in \cite{decaypaperi} that their construction had particularly good decay behavior and that the key to this behavior was that the algebraic elements in the code matrices were from a numberfield of low degree (4 to be exact). We will now study a general version of their construction and study under which conditions it is possible to achieve  the full rank condition under the extra condition that the coefficients of the codematrices are from a degree 4 number field.

Let us now suppose we have a complex quadratic field $K$ and a degree two extension $L/K$ and denote the Galois group $G(L/K)$ with $<\sigma>$.

If $a$ and $b$ are non zero elements from $L$, we can  define an embedding $\psi_{a,b}:\OO_K \mapsto \C^2 $ where
$$
\psi_{a,b}(x)=(ax, b\sigma(x)),
$$
for $x \in \OO_L$.
The following is then a standard result.

\begin{proposition}\label{singlecode}
We have that $\psi_{a, b}(\OO_L)$  is  a $4$-dimensional lattice in $\C^2$, has NVD property and is therefore
DMT optimal in $1\times 1$ MIMO channel.
\end{proposition}

Let us now suppose we have chosen elements $a,b,c$ and $d$ from the field $L$ and that the first user has code $\psi_{a,b}(\OO_L)$ and the second user has code $\psi_{c,d}(\OO_L)$. The joint codewords in the MAC code $(\psi_{a,b}(\OO_L),\psi_{c,d}(\OO_L))$ now have form
$$
\begin{pmatrix}
           ax & b \sigma(x) \\
           cy & d \sigma(y) \\
\end{pmatrix}.
$$
We are now interested in when this construction will yield codes with generalized full rank property.

We will denote with $\OO_{L}^{*}$ the ring of integers of $L$ without zero element.
\begin{thm}
\label{2olemassa}
Let $K/\Q$ and $L/K$ be two field extensions of degree $2$ and $a,b,c,d \in L$. Let also $\sigma$ be the non-trivial element in the Galois group $Gal(L/K)$. Define
$$
C = \left\{ \left(
         \begin{array}{cc}
           ax & b \sigma(x) \\
           cy & d \sigma(y) \\
         \end{array}
       \right) | x,y \in \OO_{L}^{*} \right\}.
$$
There exists a matrix in $C$ with zero determinant if and only if
\begin{equation}
\label{kaava1}
\left|
         \begin{array}{cc}
           N(a) & N(b) \\
           N(c) & N(d) \\
         \end{array}
       \right|=0,
\end{equation}
where the function $N=N_{L/K}$ denotes the norm of extension $L/K$.
\end{thm}
\begin{proof}
Assume first that we have a matrix
$$\left(
         \begin{array}{cc}
           ax & b \sigma(x) \\
           cy & d \sigma(y) \\
         \end{array}
       \right)$$
        with zero determinant. This means that $adx\sigma(y)=bc\sigma(x)y$, which gives $N(a)N(d)N(x)N(\sigma(y))=N(b)N(c)N(\sigma(x))N(y)$. Continuing we get  $N(a)N(d)=N(b)N(c)$ \emph{i.e.} $N(a)N(d)-N(b)N(c)=0$.

Assume then that  $N(a)N(d)-N(b)N(c)=0$. If $b$ or $c$ is zero then $N(a)N(d)=0$ \emph{i.e.} $a$ or $d$ is zero and then for all $x,y \in \OO_L$ we have
 $$
 \left|
         \begin{array}{cc}
           ax & b \sigma(x) \\
           cy & d \sigma(y) \\
         \end{array}
       \right| = 0.
 $$
 Otherwise $N(\frac{ad}{bc})=1$. Then by Hilbert 90 we have some $z \in L$ such that $\frac{ad}{bc}=\frac{\sigma(z)}{z}$. Then write $z=\frac{w}{n}$ with $w \in \OO_{L}^{*}$ and $n \in \Z$. This gives $\frac{ad}{bc}=\frac{\sigma(w)}{w}$ \emph{i.e.} $adw-bc\sigma(w)=0$. This means that the determinant of
$$
\left(
         \begin{array}{cc}
           aw & b \sigma(w) \\
           c1 & d \sigma(1) \\
         \end{array}
       \right)
$$
is zero.
\end{proof}

Let us now  summarize the properties of the codes of previous type.
\begin{itemize}
\item The single user codes are DMT optimal lattice codes.
\item The MAC code has generalized full rank property.
\item The overall MAC code matrices have coefficients from a number field of low degree.
\end{itemize}

In the next section we will generalize these properties for the MAC codes.
The reader can check  Equation \eqref{generalsingle} to see that when restricting to a single antenna case  our general construction has indeed the form described in this section.

\subsection{ Construction of multi access codes with several transmission antennas}\label{construction}

From now on we concentrate  on the  scenario where we have $U \in \Z_+$ users and each user has $n_t \in \Z_+$ transmission antennas.  Throughout this section we assume $K$ to be an imaginary quadratic extension of $\Q$ with class number $1$. The field  $L$ is a cyclic Galois extension of $K$ of degree $Un_t$, such that $L=K(\alpha)$ with $\alpha \in \R$, $\sigma$ a generating element in $Gal(L/K)$ and $p \in \OO_{K}$ an inert prime in $L/K$. We also define $\tau = \sigma^{U}$ and $F$ to be the fixed field of $\tau$. So we have $[L:F]=n_t$, $[F:K]=U$, $Gal(L/F)=<\tau>$, and $Gal(F/K)=<\sigma_{F}>$ where $\sigma_{F}$ is a restriction of $\sigma$ in $F$. Let  $v=v_p$ be the   $p$-adic valuation  of the field $L$. In this section, when we say that $L/K$, $p$, and $\sigma$ are suitable we mean that they are as above.

 We skip the proof of the following proposition.
\begin{proposition}\label{existence}
For every complex quadratic field $K$, having class number 1, and for any $U$ and for any $n_t$ we have   a suitable degree $n_t U$ extension $L/K$,  prime $p\in \OO_K$ and automorphism $\sigma \in Gal(L/K)$.
\end{proposition}

We are now ready to begin to build  our MAC codes.

We can define an associative $F$-algebra
$$
\mathcal{D}=(L/F,\tau, p)=L\oplus uL\oplus u^2L\oplus\cdots\oplus u^{n_t-1}L,
$$
where   $u\in\mathcal{A}$ is an auxiliary
generating element subject to the relations
$xu=u\tau(x)$ for all $x\in L$ and $u^{n_t}=p$. The choice of  $p$ and fields  $L$ and $F$ guarantees that $\mathcal{D}$ is a division algebra.

 The ${\cal O}_L$-module
$$ \Lambda={\cal O}_L\oplus u {\cal O}_L\oplus\cdots\oplus
u^{n_t-1}{\cal O}_L, $$
where ${\cal O}_L$ is the ring of integers, is a subring in the
algebra $(L/F,\tau,p)$. We refer to this ring as the {\it natural order}.

Let us now suppose  that we have an element $x\in \Lambda$. It can be written as $\sum_{i=0}^{n_t-1}x_iu^{i}$,
where $x_j \in \OO_{L}$ for all $j=1, \dots ,n_t$. We now have the left regular representation $\psi:\Lambda \mapsto M_{n_t}(\OO_L)$, where
$\psi(x)=$
\begin{equation}\label{representation}
\left(
         \begin{array}{ccccc}
           x_1   &     p \tau(x_{n_t}) & p \tau^2(x_{n_t-1}) & \ldots & p \tau^{n_t-1}(x_2) \\
           x_2   &       \tau(x_1) & p \tau^2(x_{n_t}) & \ldots & p \tau^{n_t-1}(x_3) \\
           x_3   &       \tau(x_2) &   \tau^2(x_1) & \ldots & p \tau^{n-1}(x_4) \\
           \vdots & \vdots         & \vdots          &        & \vdots              \\
           x_{n_t-1}   &   \tau(x_{n_t-2}) &   \tau^2(x_{n_t-3}) & \ldots & p \tau^{n_t-1}(x_{n_t}) \\
           x_{n_t}   &       \tau(x_{n_t-1}) &   \tau^2(x_{n_t-2}) & \ldots &   \tau^{n_t-1}(x_1) \\
         \end{array}
       \right).
\end{equation}
If the context is clear we can also use  the notation  $M(x_1,\cdots, x_{n_t})=\psi(x)$.

Let us suppose we consider $U$ user MAC scenario, where each of the  single users has $n_t$ transmit antennas.
Note that in this definition  we are using the notation $M(x_1,\dots, x_{n_t})=\psi(x_1 +x_2u+\cdots+ x_{n_t}u^{n_t-1})$.

\begin{definition}[MAC code]
Define $M_j=M(x_{j,1}, x_{j,2}, \dots, x_{j,n_t})$ for all $j=1, \dots ,U$. In our multi access system the code $\mathcal{C}_j$ of $j$th user consists of $n_t \times Un_t$ matrices $B_j=$
$$
 \left( M_j , \sigma(M_j) , \sigma^2(M_j) , \dots , p^{-m} \sigma^{j-1}(M_j) , \dots , \sigma^{U-1}(M_j) \right)
$$
where $m$ is any rational integer strictly greater than $\frac{U(n_t-1)}{2}$ and $x_{j,l} \neq 0$ for some $l$. Here  $m$ is same for all the users. Then the whole code $C_{U,n_t}$ consists of matrices

$$
A = \left(
         \begin{array}{c}
           B_1   \\
           B_2    \\
           \vdots              \\
           B_U \\
         \end{array}
       \right),
$$
where $B_j \in \mathcal{C}_j$ for all $j=1, \dots, U$. This means that the matrices $A \in C_{U,n_t}$ have form
\begin{equation}\label{code}
\left(
         \begin{array}{cccc}
           p^{-m} M_1   &  \sigma(M_1) & \ldots &  \sigma^{U-1}(M_1) \\
           M_2   & p^{-m} \sigma(M_2) & \ldots &  \sigma^{U-1}(M_2) \\
           \vdots & \vdots                  &        & \vdots              \\
           M_U   &  \sigma(M_U) & \ldots & p^{-m} \sigma^{U-1}(M_U) \\
         \end{array}
       \right).
\end{equation}
\end{definition}

The code depends on how we did choose $L/K$, $p$, $\sigma$, and $m$, so to be precise, we can also refer to $C_{U,n_t}$ with $C_{U,n_t}(L/K,p,\sigma,m)$. Let us call the family of all such codes $C_{U,n_t}(L/K,p,\sigma,m)$ (\emph{i.e.} codes constructed with any suitable $L/K$, $p$, $\sigma$, and $m$) by $\mathfrak{C}_{U,n_t}$. That is
$$
\mathfrak{C}_{U,n_t} = \bigcup_{L/K,p,\sigma,m} \{ C_{U,n_t}(L/K,p,\sigma,m) \}
$$
where $L/K$, $p$, $\sigma$, and $m$ are any suitable ones.

According to Proposition \ref{existence}  we can always find suitable $L/K$, $p$, $\sigma$, and $m$ for any $U \in \Z_+$ and $n_t \in \Z_+$. We therefore  have the following theorem.
\begin{thm}
For any choice of $U \in \Z_+$ and $n_t \in \Z_+$ we have $\mathfrak{C}_{U,n} \neq \emptyset$.
\end{thm}

We will skip the proof of the following proposition, stating that each  of the single user codes satisfies the NVD condition and are therefore DMT optimal as a single user MIMO code.
\begin{proposition}\label{1user}
Let $C_{U,n_t} \in \mathfrak{C}_{U,n_t}$ and $\mathcal{C}_j$ be the $j$th users code in the system $C_{U,n_t}$ for some $j \in {1, \dots, U}$. Then the code $\mathcal{C}_j$ is a $2Un_t^2$-dimensional lattice code    with the NVD property.
\end{proposition}

Note that the code $C_{U,1}=C_{U,1}(L/K,p,\sigma,1) \in \mathfrak{C}_{U,1}$, a code for $U$ users each having one transmission antenna, consists of matrices of form

\begin{equation}\label{generalsingle}
\left(
         \begin{array}{ccccc}
           p^{-1}x_1   &         \sigma(x_1) &       \sigma^2(x_1) & \ldots &         \sigma^{U-1}(x_1) \\
                 x_2   & p^{-1}  \sigma(x_2) &       \sigma^2(x_2) & \ldots &         \sigma^{U-1}(x_2) \\
                 x_3   &         \sigma(x_3) & p^{-1}\sigma^2(x_3) & \ldots &         \sigma^{U-1}(x_3) \\
           \vdots      & \vdots              & \vdots              &        & \vdots                    \\
                 x_U   &         \sigma(x_U) &       \sigma^2(x_U) & \ldots & p^{-1}  \sigma^{U-1}(x_U) \\
         \end{array}
       \right).
\end{equation}

Note also that the code $C_{1,n_t}=C_{1,n_t}(L/K,p,\sigma,m) \in \mathfrak{C}_{1,n_t}$ is a usual single user  code multiplied by $p^{-m}$.

Let us now prove that the defined MAC code satisfies the generalized rank criterion. We need first the following Lemma.

\begin{lemma}
\label{determinanttivaluaatio}
Let $x_j \in \OO_{L}$ for all $j=1, \dots , n_t$ such that $x_l \neq 0$ for some $l$ and $\min(v(x_1), \dots, v(x_{n_t}))=0$. We then have
$$
\det(M(x_1, x_2, \dots, x_{n_t})) \neq 0
$$
and
$$
v(\det(M(x_1, x_2, \dots, x_{n_t}))) \leq n_t - 1.
$$
\end{lemma}
\begin{proof}
The first inequality follows as  according to equation \eqref{representation} $M(x_1, x_2, \dots, x_{n_t})$ is a matrix representation of a nonzero element in a division algebra $\D$ and determinant of $M(x_1, x_2, \dots, x_{n_t})$ is the reduced norm.
Write $M=M(x_1, x_2, \dots, x_{n_t})$ and $N=N_{L/F}$. Assume first that $v(x_1)=0$. Then the determinant is $N(x_1) + py$ for some $y \in \OO_L$ and hence we have $v(\det(M))=\min(v(N(x_1)),v(py))=0$. Assume then that $v(x_1),v(x_2), \dots, v(x_{l-1}) > 0$ and $v(x_l)=0$ with $1 < l \leq n_t$. Notice that in this case all the other elements $a$ of matrix $M$, than those in the left lower corner  block of side length $n_t -l+1$, have $v(a)>0$. Either they have coefficient $p$ or they are automorphic images of elements $x_1 , x_2 , \dots , x_{l-1}$. Now $\det(M) = \pm p^{l-1}N(x_l) + p^{l}z$ for some $z \in \OO_L$ since all the other terms except $\pm p^{l-1}N(x_l)$ have at most $n_t -l$ factors from this left lower corner and hence at least $n_t-(n_t-l)=l$ terms have factor $p$. This gives that $v(\det(M))=\min(v(p^{l-1}N(x_l)),v(p^{l}z))=l-1 \leq n_t -1$.
\end{proof}

\begin{thm}
Let $C_{U,n_t} \in \mathfrak{C}_{U,n_t}$. The code $C_{U,n_t}$ is a full rate code and   satisfies the generalized rank criterion.
\end{thm}
\begin{proof}
Let $A \in C_{U,n_t}=C_{U,n_t}(L/K,p,\sigma,m)$. We may assume that $\min(v(x_{j,1}), \dots, v(x_{j,n_t}))=0$, for all $j=1, \dots, U$, because otherwise we can divide extra $p$'s off. That does not have any impact on whether $\det(A)=0$ or not. The determinant of $A$ is
$$
p^{-mUn_t} \prod_{l=1}^{U} \det(\sigma^{l-1}(M_l)) + y
$$
where $v(y) \geq -m(Un_t-2)$. We know that $v(\sigma^{l-1}(\det(M_l)))=v(\det(M_l))$ because $p$ is from $K$, \emph{i.e.} from the fixed field of $\sigma$, and $\det(M_l) \neq 0$ for all $l$. Therefore
$$
v(p^{-mUn_t} \prod_{l=1}^{U} \det(\sigma^{l-1}(M_l))) = -kUn + \sum_{l=1}^{U} v(\det(M_l))
$$
that is less or equal to $-mUn_t + U(n_t-1) = U(n_t-1-mn_t)$ by \ref{determinanttivaluaatio}. But if we would have $\det(A)=0$ then
$$
v(y) = v(p^{-mUn_t} \prod_{l=1}^{U} \det(\sigma^{l-1}(M_l)))
$$
and hence $v(y) \leq U(n_t-1-mn_t)$ implying $-m(Un_t-2) \leq U(n_t-1-mn_t)$. This gives $2m \leq U(n_t-1)$ \emph{i.e.} $m \leq \frac{U(n_t-1)}{2}$ a contradiction.
\end{proof}

\begin{remark}
Using  multiblock codes from division algebras as single user codes in the MIMO MAC scenario has been done before for example in \cite{bb}, \cite{matriisi} and \cite{LuHoVeLa}. In \cite{bb} the full rank condition for codes with $n_t>1$ is achieved by using transcendental elements.
In  \cite{matriisi} the same effect is achieved with algebraic elements of high degree.

\end{remark}

\subsection{Examples}\label{examples}
Let us now give a few examples of our general code constructions.
In Table 1 we have  collected  some examples of suitable fields $K$ and $L$ and inert primes $p$,  fulfilling the conditions of Proposition \ref{existence}. If $K=\Q(i)$ then   $p_i\OO_K$ refers to the inert prime and if $K=\Q(\sqrt{-3})$ then $p_{\sqrt{-3}}\OO_K$ is inert. The inert primes and fields $L$ are found by looking at totally real subfields of  $\Q(\zeta_h)/\Q$ and then composing  them  with the field $K$.

\begin{table}[h!]\label{table1}\caption{}
$$
\begin{array}{ccccc}
\hline
           [L:K] & L                                                              & p_i & p_{\sqrt{-3}} \\
\hline
           3     & K(\zeta_7+\zeta_{7}^{-1})                                      & 2+i & \sqrt{-3} \\
           4     & K(\zeta_{17}+\zeta_{17}^{4}+\zeta_{17}^{-4}+\zeta_{17}^{-1})   & 2+i & \sqrt{-3}\\
           5     & K(\zeta_{11}+\zeta_{11}^{-1})                                  & 1+i & 2+\sqrt{-3}\\
           6     & K(\zeta_{13}+\zeta_{13}^{-1})                                  & 1+i & 2+\sqrt{-3}\\
           7     & K(\zeta_{29}+\zeta_{29}^{12}+\zeta_{29}^{-12}+\zeta_{29}^{-1}) & 1+i & \sqrt{-3}\\
\hline
\end{array}
$$
\end{table}

We get a code $C_{3,1}=C_{3,1}(\Q(i,\zeta_7+\zeta_{7}^{-1}),2+i,\sigma,1)$ \emph{i.e.} 3-user code with each user having 1 antenna by setting $L=K(\zeta_7+\zeta_{7}^{-1})$, $K=\Q(i)$, $p=2+i$, and $Gal(L/K)=<\sigma>$. Now the actual code consists of matrices
$$
\left(
         \begin{array}{ccc}
          p^{-1} x   &  \sigma(x)       & \sigma^2(x)         \\
           y         & p^{-1} \sigma(y) & \sigma^2(y)         \\
           z         &  \sigma(z)       & p^{-1} \sigma^2(z)  \\
         \end{array}
       \right)
$$
where $x,y,z \in \OO_{L}^{*}$.

We get a code $C_{2,2}=C_{2,2}(\Q(\sqrt{-3},\zeta_{17}+\zeta_{17}^{4}+\zeta_{17}^{-4}+\zeta_{17}^{-1}),\sqrt{-3},\sigma,2)$ \emph{i.e.} 2-user code with each user having 2 antennas by setting $L=K(\zeta_{17}+\zeta_{17}^{4}+\zeta_{17}^{-4}+\zeta_{17}^{-1})$, $K=\Q(\sqrt{-3})$, $p=\sqrt{-3}$, $m=2 > \frac{U(n_t-1)}{2}$, and $Gal(L/K)=<\sigma>$. Now the actual code consists of matrices
$$
\left(
         \begin{array}{cccc}
          p^{-2} x_1   &     p^{-1} \sigma^2(x_2)  &  \sigma(x_1)        &  p\sigma^3(x_2) \\
          p^{-2} x_2   &     p^{-2}  \sigma^2(x_1) &  \sigma(x_2)        &  \sigma^3(x_1) \\
           y_1         &     p \sigma^2(y_2)       &  p^{-2} \sigma(y_1) & p^{-1} \sigma^3(y_2) \\
           y_2         &       \sigma^2(y_1)       &  p^{-2} \sigma(y_2) & p^{-2}  \sigma^3(y_1) \\
         \end{array}
       \right)
$$
where $x_1,x_2,y_1,y_2 \in \OO_L$ and $x_1 \neq 0$ or $x_2 \neq 0$ and $y_1 \neq 0$ or $y_2 \neq 0$.

\section{On the decay function of codes in $\mathfrak{C}_{U,n_t}$}\label{decayanalysis}
In this section we will prove an asymptotic lower bound for the decay function of codes from $\mathfrak{C}_{U,n_t}$. In \cite{decaypaperi} the authors give a general asymptotic upper bound for a decay function in the case that only one user is properly using the code \emph{i.e.} $N_1$ can be anything but $N_2= \dots =N_U=1$ are restricted. We will see that in this special case our codes have asymptotically the best possible decay.

\begin{lemma}\cite[Theorem 7.8.8]{HORN}
If $A,B \in M_n(\C)$ are positive definite, we have
$$
(\det(A+B))^{1/n}\geq\det(A)^{1/n} +\det(B)^{1/n},
$$
where $n$ is an integer.
\end{lemma}

\begin{lemma}\label{blocks}
Let us suppose that $F$ is an algebraic number field,  $\D=(L/F,\tau, \gamma)$ an index $n$ $F$-central division algebra and that   $\psi$ is a left regular representation of $\D$.
If $A$ is block matrix
$$
A=
\begin{pmatrix}
A_{1,k}&\cdots& A_{1,k}\\
\vdots&  &  \vdots \\
A_{n,1}&\cdots& A_{k,k}
\end{pmatrix}
$$
where $A_{i,j}=\psi(x_{i,j})$ for some elements $x_{i, j}\in D$, we then have that
$$
\det(A)\in F.
$$
\end{lemma}
\begin{proof}
The matrix $A$ can be considered as an element in the central simple algebra $M_k(\D)$ and the determinant  is then simply the reduced norm of this element. The theory of central simple  algebras  gives that reduced norm of any element of $M_k(\D)$ belongs  to the center.
\end{proof}

In the following we use the notation of \eqref{code}.
\begin{lemma}\label{inF}
Let $C_{U,n_t}=C_{U,n_t}(L/K,p,\sigma,m) \in \mathfrak{C}_{U,n_t}$, $A \in C_{U,n_t}$, and let $F$ be the center of the used division algebra $\D=(L/F, \tau, p)$.  Then  for  any  square matrix $A$  that consists  of blocks of form $\sigma_j(M_i)$ or $p^{-m}\sigma_j(M_i)$ we have that $\det(A) \in F$.
\end{lemma}
\begin{proof}
Each of the matrices $M_i$ in \eqref{code} are left regular representation $\psi$ of elements $x$ in $\D$. If $\psi(x)=M_i$, then
$\psi(p^{-m}x)= p^{-m}M_i$. The isomorphisms   $\sigma^j$ are elements in $Gal(L/\Q)$ and we can therefore see that for any  pair $i,j$  we have  $p^{-m}\sigma_j(M_i)=\psi(y_{1})$ and $\sigma_j(M_i)=\psi(y_{2})$, for some $y_1, y_2 \in D$. The final result is then  a direct consequence  of  Lemma \ref{blocks}.

\end{proof}

For the next theorem we need few definitions. Let $p(x)=p_0+p_1x+\dots+p_l x^l \in \Z[x]$ be a polynomial. Then we say that $H(p(x))=\max\{|p_j|\}$ is the height of the polynomial $p(x)$ and for an algebraic number $\alpha$ we define $H(\alpha)=H(\phi_{\alpha})$ where $\phi_{\alpha}$ is the minimal polynomial of $\alpha$.
The next generalization of Liouville's theorem can be found from \cite[p. 31]{shidlovskii}.
\begin{thm}
\label{algebraicapprox}
Let $\alpha \in \R$ be an algebraic number of degree $\kappa$, $H(\alpha) \leq h$, $H(P) \leq H$ and $\deg(P(x))=l \in \Z^{+}$. Then either $P(\alpha)=0$ or
$$
|P(\alpha)| \geq \frac{c^{l}}{H^{\kappa-1}}
$$
with $c=\frac{1}{3^{\kappa-1}h^{\kappa}}$.
\end{thm}

Now we are ready to give a  lower bound for the decay function of our codes. The  proof can be seen as an extension to the analysis given for BB-code in \cite{decaypaperi}.
\begin{thm}
\label{inertdecay3}
For a code $C_{U,n_t} \in \mathfrak{C}_{U,n_t}$ there exists constant $K>0$ such that
$$
D_{\I_u} (N_{i_1}, \dots, N_{i_u}) \geq \left\{
  \begin{array}{l l}
    \frac{K}{(N_{i_1}\dots N_{i_u})^{(U-1)n_t}} & \quad \text{if $u>1$}\\
    K & \quad \text{if $u=1$}
  \end{array} \right.
$$
Especially
$$
D_{\I_u} (N) \geq \left\{
  \begin{array}{l l}
    \frac{K}{N^{u(U-1)n_t}} & \quad \text{if $u>1$}\\
    K & \quad \text{if $u=1$}
  \end{array} \right.
$$
\end{thm}

\begin{proof}
If $u=1$ then the claim is true by Proposition \ref{1user}. Next we will prove the claim in detail in the case that $u=U$ and at the end of the proof it is explained how the proof is extended for the case $1<u<U$.

Let $C_{U,n_t}=C_{U,n_t}(L/K,p,\sigma,m)$. Field extension $L/\Q$ has a basis $S_1 \cup S_2$ where $S_1=\{1, \delta, \delta^{2}, \ldots, \delta^{U-1}, \beta, \beta \delta, \beta \delta^{2}, \ldots, \beta \delta^{U-1}\}$ is a basis of $F/\Q$ with $\delta \in \R$, $K=\Q(\beta)$ and $\beta=\sqrt{-w}$ for some positive integer $w$. Notice that if $L=F$ then $S_2 = \emptyset$.

 The ring $\OO_L$ has a $\Z$-basis $\{ \gamma_{1}, \ldots, \gamma_{2Un_t} \}$. Each of these basis elements can be presented as
$$
\gamma_{l}=\sum_{a \in S_1 \cup S_2} s_{l,a} a,
$$
where $s_{l,a} \in \Q$ for all $l = 1, \dots, 2Un_t$ and $a \in S_1 \cup S_2$.

Let $A \in C_{U,n}$ be
$$
\left(
         \begin{array}{ccccc}
           p^{-m} M_1   &  \sigma(M_1) & \ldots &  \sigma^{U-1}(M_1) \\
           M_2   & p^{-m} \sigma(M_2) & \ldots &  \sigma^{U-1}(M_2) \\
           \vdots & \vdots          &        & \vdots              \\
           M_U   &  \sigma(M_U) & \ldots & p^{-m} \sigma^{U-1}(M_U) \\
         \end{array}
       \right)
$$
and $M_j=M(x_{j,1}, x_{j,2}, \dots, x_{j,n_t})$ be
$$
\left(
         \begin{array}{ccccc}
           x_{j,1}   &     p \tau(x_{j,n_t}) & p \tau^2(x_{j,n_t-1}) & \ldots & p \tau^{n_t-1}(x_{j,2}) \\
           x_{j,2}   &       \tau(x_{j,1}) & p \tau^2(x_{j,n_t}) & \ldots & p \tau^{n_t-1}(x_{j,3}) \\
           x_{j,3}   &       \tau(x_{j,2}) &   \tau^2(x_{j,1}) & \ldots & p \tau^{n_t-1}(x_{j,4}) \\
           \vdots & \vdots         & \vdots          &        & \vdots              \\
           x_{j,n_t}   &       \tau(x_{j,n_t-1}) &   \tau^2(x_{j,n_t-2}) & \ldots &   \tau^{n_t-1}(x_{j,1}) \\
         \end{array}
       \right)
$$
as usual.

Now for any $j=0, \dots, Un_t-1$ we have
$$
\sigma^{j}(x_{i,h}) = \sum_{l=1}^{2Un_t} u_{i,h,l} \sigma^{j}(\gamma_l)
$$
where $u_{i,h,l} \in \Z$ and $|u_{i,h,l}| \leq N_{i}$ for all $i$, $h$ and $l$.

Then the determinant $\det(A)$ is a sum consisting of $Un_t!$ elements of form
$$
p^{-f} \prod_{j=0}^{Un_t}  \sigma^{j}(x_{i_j,h_j}) = p^{-f} \prod_{j=0}^{Un_t} ( \sum_{l=1}^{2Un_t} u_{i_j,h_j,l} \sigma^{j}(\gamma_l) )
$$
where $f \leq kUn_t$ and $i_j$ gets exactly $n_t$ times all the values $1, \dots, U$ and $h_j$ gets values from $\{ 1, \dots, n_t \}$.

Now substituting $\gamma_{l}=\sum_{a \in S_1 \cup S_2} s_{l,a} a$ gives that the determinant is a sum consisting of elements of form
$$
p^{-f} \prod_{j=0}^{Un_t} ( \sum_{l=1}^{2Un_t} u_{i_j,h_j,l} \sum_{a \in S_1 \cup S_2} s_{l,a} \sigma^{j}(a) ).
$$

We also write
$$
\sigma^{j}(a)=\sum_{a \in S_1 \cup S_2} t_{j,a} a
$$
where $t_{j,a} \in \Q$ for all $j,a$ and find that $p^{-f} \prod_{j=0}^{Un_t} ( \sum_{l=1}^{2Un_t} u_{i_j,h_j,l} \sum_{a \in S_1 \cup S_2} s_{l,a} \sigma^{j}(a) )$ can be written as a sum of elements of form
$$
K_{1} p^{-f} \sum_{a \in S_1 \cup S_2} u_a a
$$
where $K_{1} \in \Q$ is some constant, $u_a \in \Z$, and $u_a = \OO((N_1 \dots N_U)^{n_t})$.

Writing also $p$ using basis $S_1 \cup S_2$ we see that the whole determinant $\det(A)$ can be written as a sum of elements of form
$$
\sum_{a \in S_1 \cup S_2} u_{a}' a
$$
multiplied by some constant $K_{2}$ and here we have $u_{a}' \in \Z$, and $u_{a}' = \OO((N_1 \dots N_U)^{n_t})$.

On the other hand we know that $\det(A) \in F$ so by uniqueness of basis representation we know that $\det(A)$ is a sum consisting of elements of form
$$
\sum_{a \in S_1} u_{a}' a = \sum_{i=0}^{U-1} u_{\delta^{i}}' \delta^{i} + \beta \sum_{i=0}^{U-1} u_{\delta^{i}\beta}' \delta^{i}
$$
and hence
$$
|\det(A)| = |K_{2}| |\sum_{l=0}^{U-1} H_{l} \delta^{l} + \beta \sum_{l=0}^{U-1} J_{l} \delta^{l}|,
$$
where $H_l , J_l \in \Z$ and $|H_l| , |J_l|$ are of size $\OO((N_1 \cdots N_U)^{n_t})$ for all $l = 0, \dots, U-1$.

Using the fact that $\delta$ is real we get
$$
|\det(A)| \geq \frac{K_{2}}{2} (|\sum_{l=0}^{U-1} H_{l} \delta^{l}| + |\sum_{l=0}^{U-1} J_{l} \delta^{l}|).
$$

Now using \ref{algebraicapprox} and noticing that $\deg(\delta) = U$ we have
$$
|\det(A)| \geq \frac{K}{(N_{1} \cdots N_{U})^{(U-1)n_t}},
$$
where $K$ is some positive constant.

Assume now that $1<u<U$. Write $M=M(X_{i_1},\dots,X_{i_u})$. It is well known that
$$
\text{rank}(MM^{\dagger})
=  \text{rank}(M)
=  un_t
$$
and hence the determinant of $MM^{\dagger}$ is nonzero.

By Lemma \ref{inF} we see that $\det(MM^{\dagger}) \in F$. $MM^{\dagger}$ is a matrix where the element in the place $(\hat{i},\hat{j})$ where $(i-1)n_t < \hat{i} \leq in_t$ and $(j-1)n_t < \hat{j} \leq jn_t$, is of size $\OO(N_i N_j)$. Using these facts we get the wanted result.

\end{proof}

\begin{corollary}\label{pigeon2}
For a code $C_{U,n_t} \in \mathfrak{C}_{U,n_t}$ there exists constants $k>0$ and $K>0$ such that
$$
\frac{k}{N^{(U-1)n_t}} \leq D(N_1=N, N_2= \ldots =N_U=1) \leq \frac{K}{N^{(U-1)n_t}}.
$$
\end{corollary}

\section{Appendix}\label{proofs}
 In this section we will give proofs of some results that were earlier postponed.

\begin{lemma}
\label{matriisitulo}
Let $\mathbf{c}_1, \mathbf{c}_2, \dots, \mathbf{c}_k, \mathbf{e}_1, \mathbf{e}_2, \dots, \mathbf{e}_{k-1} \in \C^{n}$, and $\mathbf{c}_i-\mathbf{e}_i \in \R(\mathbf{c}_{i+1}, \mathbf{c}_{i+2}, \dots, \mathbf{c}_{k})$ for $i=1, \dots, k-1$. Write also
$$
A = \left(
         \begin{array}{c}
           \mathbf{c}_1          \\
           \mathbf{c}_2          \\
           \vdots       \\
           \mathbf{c}_{k-1}          \\
           \mathbf{c}_k          \\
         \end{array}
       \right)
\,\,\mathrm{and}\,\,
B = \left(
         \begin{array}{c}
           \mathbf{e}_1          \\
           \mathbf{e}_2          \\
           \vdots       \\
           \mathbf{e}_{k-1}      \\
           \mathbf{c}_k          \\
         \end{array}
       \right).
$$

Then we have $\det(AA^{\dag})=\det(BB^{\dag})$.
\end{lemma}
\begin{proof}
If $k>n$ then $\det(AA^{\dag})=0=\det(BB^{\dag})$. If $k=n$ then $\det(A)$ is
$$
\left|
         \begin{array}{c}
           \mathbf{c}_1          \\
           \mathbf{c}_2          \\
           \vdots       \\
           \mathbf{c}_{k-1}          \\
           \mathbf{c}_k          \\
         \end{array}
       \right|
 = \left|
         \begin{array}{c}
           \mathbf{e}_1          \\
           \mathbf{c}_2          \\
           \vdots       \\
           \mathbf{c}_{k-1}          \\
           \mathbf{c}_k          \\
         \end{array}
       \right|
 = \left|
         \begin{array}{c}
           \mathbf{e}_1          \\
           \mathbf{e}_2          \\
           \vdots       \\
           \mathbf{c}_{k-1}          \\
           \mathbf{c}_k          \\
         \end{array}
       \right|
=\dots
 = \left|
         \begin{array}{c}
           \mathbf{e}_1          \\
           \mathbf{e}_2          \\
           \vdots       \\
           \mathbf{e}_{k-1}          \\
           \mathbf{c}_k          \\
         \end{array}
       \right|
$$
i.e. $\det(B)$ and hence $\det(AA^{\dag})=\det(BB^{\dag})$.

Assume $k<n$. Let $\mathbf{v}_1, \dots, \mathbf{v}_{n-k} \in \C^{n}$ be such that $\mathbf{v}_1 \in \R(\mathbf{c}_1, \mathbf{c}_2, \dots, \mathbf{c}_k,)^{\bot}\setminus \{\mathbf{0}\}$, $\mathbf{v}_2 \in \R(\mathbf{v}_1, \mathbf{c}_1, \mathbf{c}_2, \dots, \mathbf{c}_k)^{\bot}\setminus \{\mathbf{0}\}$, ..., $\mathbf{v}_{n-k} \in \R(\mathbf{v}_1, \mathbf{v}_2, \dots \mathbf{v}_{n-k-1}, \mathbf{c}_1 , \mathbf{c}_2, \dots, \mathbf{c}_k)^{\bot}\setminus \{\mathbf{0}\}$. Now (as in the case $n=k$) we have
$$
\det(\left(
         \begin{array}{c}
           \mathbf{c}_1          \\
           \vdots       \\
           \mathbf{c}_{k-1}      \\
           \mathbf{c}_k          \\
           \mathbf{v}_1          \\
           \vdots       \\
           \mathbf{v}_{n-k}      \\
           \end{array}
       \right))
=
\det(\left(
         \begin{array}{c}
           \mathbf{e}_1     \\
           \vdots       \\
           \mathbf{e}_{k-1}          \\
           \mathbf{c}_k          \\
           \mathbf{v}_1          \\
           \vdots       \\
           \mathbf{v}_{n-k}      \\
           \end{array}
       \right))
$$
and hence
$$
\left|
         \begin{array}{cccccc}
           \mathbf{c}_1 \mathbf{c}_{1}^{*}          & \dots   & \mathbf{c}_1 \mathbf{c}_{k}^{*}     & \mathbf{c}_1 \mathbf{v}_{1}^{*}     & \dots & \mathbf{c}_1 \mathbf{v}_{n-k}^{*}       \\
           \vdots                 &         & \vdots            & \vdots            &       & \vdots  \\
           \mathbf{c}_k \mathbf{c}_{1}^{*}          & \dots   & \mathbf{c}_k \mathbf{c}_{k}^{*}     & \mathbf{c}_k \mathbf{v}_{1}^{*}     & \dots & \mathbf{c}_k \mathbf{v}_{n-k}^{*}  \\
           \mathbf{v}_1 \mathbf{c}_{1}^{*}          & \dots   & \mathbf{v}_1 \mathbf{c}_{k}^{*}     & \mathbf{v}_1 \mathbf{v}_{1}^{*}     & \dots & \mathbf{v}_1 \mathbf{v}_{n-k}^{*}  \\
           \vdots                 &         & \vdots            & \vdots            &       & \vdots           \\
           \mathbf{v}_{n-k} \mathbf{c}_{1}^{*}      & \dots   & \mathbf{v}_{n-k} \mathbf{c}_{k}^{*} & \mathbf{v}_{n-k} \mathbf{v}_{1}^{*} & \dots & \mathbf{v}_{n-k} \mathbf{v}_{n-k}^{*}  \\
           \end{array}
       \right|
$$
is equal to
$$
\left|
         \begin{array}{cccccc}
           \mathbf{e}_1 \mathbf{e}_{1}^{*}          & \dots   & \mathbf{e}_1 \mathbf{c}_{k}^{*}     & \mathbf{e}_1 \mathbf{v}_{1}^{*}     & \dots & \mathbf{e}_1 \mathbf{v}_{n-k}^{*}       \\
           \vdots                 &         & \vdots            & \vdots            &       & \vdots  \\
           \mathbf{c}_k \mathbf{e}_{1}^{*}          & \dots   & \mathbf{c}_k \mathbf{c}_{k}^{*}     & \mathbf{c}_k \mathbf{v}_{1}^{*}     & \dots & \mathbf{c}_k \mathbf{v}_{n-k}^{*}  \\
           \mathbf{v}_1 \mathbf{e}_{1}^{*}          & \dots   & \mathbf{v}_1 \mathbf{c}_{k}^{*}     & \mathbf{v}_1 \mathbf{v}_{1}^{*}     & \dots & \mathbf{v}_1 \mathbf{v}_{n-k}^{*}  \\
           \vdots                 &         & \vdots            & \vdots            &       & \vdots           \\
           \mathbf{v}_{n-k} \mathbf{c}_{1}^{*}      & \dots   & \mathbf{v}_{n-k} \mathbf{c}_{k}^{*} & \mathbf{v}_{n-k} \mathbf{v}_{1}^{*} & \dots & \mathbf{v}_{n-k} \mathbf{v}_{n-k}^{*}  \\
           \end{array}
       \right|.
$$
And since the way we chose $\mathbf{v}_1, \dots, \mathbf{v}_{n-k}$ this means that
$$
\left|
         \begin{array}{cccccc}
           \mathbf{c}_1 \mathbf{c}_{1}^{*}          & \dots       & \mathbf{c}_1 \mathbf{c}_{k}^{*}     & 0     & \dots & 0       \\
           \vdots                 &                      & \vdots            & \vdots            &       & \vdots  \\
           \mathbf{c}_k \mathbf{c}_{1}^{*}          & \dots       & \mathbf{c}_k \mathbf{c}_{k}^{*}     & 0     & \dots & 0  \\
           0          & \dots       & 0     & \mathbf{v}_1 \mathbf{v}_{1}^{*}     & \dots &0  \\
           \vdots                 &                     & \vdots            & \vdots            &       & \vdots           \\
           0      & \dots  & 0 & 0 & \dots & \mathbf{v}_{n-k} \mathbf{v}_{n-k}^{*}  \\
           \end{array}
       \right|
$$
is equal to
$$
\left|
         \begin{array}{cccccc}
           \mathbf{e}_1 \mathbf{e}_{1}^{*}          & \dots       & \mathbf{e}_1 \mathbf{c}_{k}^{*}     & 0     & \dots & 0       \\
           \vdots                 &                     & \vdots            & \vdots            &       & \vdots  \\
           \mathbf{c}_k \mathbf{e}_{1}^{*}          & \dots      & \mathbf{c}_k \mathbf{c}_{k}^{*}     & 0     & \dots & 0  \\
           0          & \dots  & 0         & \mathbf{v}_1 \mathbf{v}_{1}^{*}     & \dots & 0  \\
           \vdots                 &                  & \vdots            & \vdots            &       & \vdots           \\
           0      & \dots  & 0 & 0 & \dots & \mathbf{v}_{n-k} \mathbf{v}_{n-k}^{*}  \\
           \end{array}
       \right|
$$
because if we write $\mathbf{e}_i = \mathbf{c}_i - \mathbf{x}_i$ where $\mathbf{x}_i \in L(\mathbf{c}_{i+1}, \dots, \mathbf{c}_k)$ then $\mathbf{v}_j \mathbf{e}_{i}^{*} = \mathbf{v}_j (\mathbf{c}_i - \mathbf{x}_i)^{*} = \mathbf{v}_j \mathbf{c}_{i}^{*} - \mathbf{v}_j \mathbf{x}_{i}^{*} = 0 - 0 = 0$ for all $i=1, \dots, n-1$ and $j=1, \dots, n-k$.
This gives that
$$
|\mathbf{v}_1|^2 \dots |\mathbf{v}_{n-k}|^2 \det(AA^{\dag}) = |\mathbf{v}_1|^2 \dots |\mathbf{v}_{n-k}|^2 \det(BB^{\dag})
$$
and hence $ \det(AA^{\dag}) = \det(BB^{\dag})$.
\end{proof}

\begin{proof}[The proof of Theorem \ref{maindecay}]
Let us use the notation $C_l=(\mathbf{c}_{l,1}^{\top} , \dots , \mathbf{c}_{l,n_{t}}^{\top})^{\top}$ for $l=i_1, \dots, i_u$.

Let us first fix some small $C_{i_u} \in \mathbf{L}_{i_u}(N_{i_u})$. Now $|C_{i_u}|=\OO(1)$. Then write
$$
W_{i_u} = \{ (\mathbf{x}_{1}^{\top} , \dots , \mathbf{x}_{n_{t}}^{\top})^{\top} | \mathbf{x}_i \in \R(\mathbf{c}_{i_u,1} , \dots , \mathbf{c}_{i_u,n_{t}}) \}.
$$
Then let $\pi_{W_{i_u}^{\bot}}: M_{n_{t} \times k}(\C) \rightarrow W_{i_u}^{\bot}$ be an orthogonal projection as before.

A subspace $W_{i_u}^{\bot}$ has $\dim_{\R}(W_{i_u}^{\bot})=2n_{t}k-\dim_{\R}(W_{i_u})=2n_{t}k-2n_{t}^2=2n_{t}(k-n_{t})$. By corollary \ref{latticepigprincip} we have some $C_{i_{u-1}} \in \mathbf{L}_{i_{u-1}}(N_{i_{u-1}})$ such that
$$
||\pi_{W_{i_u}^{\bot}}(C_{i_{u-1}})||=\OO(N_{i_{u-1}}^\frac{\dim_{\R}(W_{i_u}^{\bot})-\dim_{\R}(\mathbf{L}_{i_{u-1}})}{\dim_{\R}(W_{i_u}^{\bot})})=\OO(N_{i_{u-1}}^{-\frac{n_{t}}{k-n_{t}}}).
$$

Now similarly build
$W_{i_{u-l}}$ for $l=0, \dots, u-2$ by setting
\begin{equation}
\begin{split}
W_{i_{u-l}} = \{ & (\mathbf{x}_{1}^{\top} , \dots , \mathbf{x}_{n_{t}}^{\top})^{\top} | \mathbf{x}_i \in \R(\mathbf{c}_{i_u,1} , \dots , \mathbf{c}_{i_u,n_{t}} , \mathbf{c}_{i_{u-1},1}, \\
& \dots , \mathbf{c}_{i_{u-1},n_{t}} , \dots , \mathbf{c}_{i_{u-l},1} , \dots , \mathbf{c}_{i_{u-l},n_{t}}) \}
\end{split}
\end{equation}

This gives $\dim_{\R}(W_{i_{u-l}}^{\bot})=2n_{t}k-\dim_{\R}(W_{i_{u-l}})=2n_{t}k-2n_{t}^2(l+1)=2n_{t}(k-n_{t}l-n_{t})$. And again, by corollary \ref{latticepigprincip} we find $C_{i_{u-l-1}} \in \mathbf{L}_{i_{u-l-1}}(N_{i_{u-l-1}})$ such that
\begin{eqnarray*}
||\pi_{W_{i_{u-l}}^{\bot}}(C_{i_{u-l-1}})||&=&\OO(N_{i_{u-l-1}}^\frac{\dim_{\R}(W_{i_{u-l}}^{\bot})-\dim_{\R}(\mathbf{L}_{i_{u-1}})}{\dim_{\R}(W_{i_{u-l}}^{\bot}})\\
&=&\OO(N_{i_{u-l-1}}^{-\frac{n_{t}l+n_{t}}{k-n_{t}l-n_{t}}})
\end{eqnarray*}
where $\pi_{W_{i_{u-l}}^{\bot}}: M_{n_{t} \times k}(\C) \rightarrow W_{i_{u-l}}^{\bot}$ is an orthogonal projection as before.

Lemma \ref{matriisitulo} gives that if $A=(C_{i_1}^{\top} , \dots , C_{i_u}^{\top})^{\top}$ and $B=(\pi_{i_2}(C_{i_1})^{\top} , \dots , \pi_{i_u}(C_{i_{u-1}})^{\top}, C_{i_u}^{\top})^{\top}$ then $\det(AA^{\dag})=\det(BB^{\dag})$ that is of size
$$
\OO((\prod_{l=0}^{u-2} N_{i_{u-l-1}}^{-\frac{n_{t}l+n_{t}}{k-n_{t}l-n_{t}}} )^{2n_{t}}) = \OO(\prod_{l=1}^{u-1} N_{i_l}^{-\frac{2n_{t}^2 (u-l)}{k-n_{t}(u-l)}} ).
$$

\end{proof}

\begin{proof}[The proof of Theorem \ref{thm:E_I}]
Given ${\cal I}=\{i_1, \ldots, i_u\}$, we first derive a lower bound on $d_{\min}(H_{i_1}, \ldots, H_{i_u})$. Let $H_{\cal I}=[H_{i_1} \ \cdots \ H_{i_u}]$, $A_{\cal I}=\diag\left(\kappa_{i_1} I_{n_t}, \cdots, \kappa_{i_u} I_{n_t}\right)$,  $\Delta X_{\cal I}=[ \Delta X_{i_1}^\top \ \cdots \ \Delta X_{i_u}^\top]^\top$, and $\Delta \tilde{X}_{\cal I}=A_{\cal I} \Delta X_{\cal I}$. Note that
\[
\norm{\sum_{i \in {\cal I}} \kappa_i H_i \Delta X_i}^2 \ = \
\tr \left( H_{\cal I}^\dag  H_{\cal I} \Delta \tilde{X}_{\cal I} \Delta \tilde{X}_{\cal I}^\dag \right).
\]
Let $\lambda_i$ and $\ell_j$ be respectively the eigenvalues of $H_{\cal I}^\dag  H_{\cal I}$ and $\Delta \tilde{X}_{\cal I} \Delta \tilde{X}_{\cal I}^\dag$ that are ordered as
\[
\begin{array}{l}
0 = \lambda_1 = \cdots = \lambda_{u n_t - K_u} < \lambda_{u n_t-K_u+1} \leq \cdots \leq \lambda_{u n_t}\\
\ell_1 \geq \ell_2 \cdots \geq \ell_{u n_t}
\end{array}
\]
where $K_u = \min\{ u n_t, n_r\}$. For each $t \in \Z$, $1 \leq t \leq K_u$, we have
\begin{IEEEeqnarray*}{rCl}
&& d_{\min}(H_{i_1}, \ldots, H_{i_u}) =  \min_{\Delta X_{\cal I} \neq {\bf 0}} \tr \left( H_{\cal I}^\dag  H_{\cal I} \Delta \tilde{X}_{\cal I} \Delta \tilde{X}_{\cal I}^\dag \right)\\
& \stackrel{\text{(a)}}{\geq} & \min_{\Delta X_{\cal I} \neq {\bf 0}} \sum_{i=u n_t - K_u+1}^{u n_t} \lambda_i \ell_i
 \geq  \min_{\Delta X_{\cal I} \neq {\bf 0}} \sum_{i=u n_t - t +1}^{u n_t} \lambda_i \ell_i\\
& \stackrel{\text{(b)}}{\dot\geq} &\min_{\Delta X_{\cal I} \neq {\bf 0}} \left[ \prod_{i=u n_t - t +1}^{u n_t} \lambda_i \ell_i \right]^{\frac{1}{t}}
 \stackrel{\text{(c)}}{\geq} \min_{\Delta X_{\cal I} \neq {\bf 0}} \left[ \prod_{i=u n_t - t +1}^{u n_t} \lambda_i \right]^{\frac{1}{t}} \left[
\frac{\left(\det(A_{\cal I})D_{\cal I}(N_{i_1}, \ldots, N_{i_u})\right)^2}{\prod_{i=1}^{un_t-t} \ell_i}\right]^{\frac{1}{t}}\\
& \stackrel{\text{(d)}}{\dot\geq} &\left[ \prod_{i=u n_t - t +1}^{u n_t} \lambda_i \right]^{\frac{1}{t}} \left[ \frac{\left(\det(A_{\cal I})D_{\cal I}(N_{i_1}, \ldots, N_{i_u})\right)^2}{\snr^{u n_t - t}}\right]^{\frac{1}{t}}\\
\end{IEEEeqnarray*}
where (a) follows from the mis-match inequality \cite{KW}, (b) is due to the AM-GM inequality, (c) follows from the definition of $D_{\cal I}(N_{i_1}, \ldots, N_{i_u})$ and that $\prod_{i=1}^{u n_t} \ell_i = \det\left(A_{\cal I} \Delta X_{\cal I} \Delta X_{\cal I}^\dag A_{\cal I}^\dag \right) \geq \left( \det(A_{\cal I}) D_{\cal I}(N_{i_1}, \ldots, N_{i_u})\right)^2$, and (d) is again from the AM-GM inequality and
\[
\prod_{i=1}^{un_t-t} \ell_i \dot\leq \left( \sum_{i=1}^{u n_t-t} \ell_i\right)^{u n_t-t} \leq \norm{\Delta \tilde{X}_{\cal I}}^{2 (un_t - t)} \leq \snr^{u n_t - t}.
\]

Set $\lambda_{i+u n_t-K_u} = \snr^{-\alpha_i}$ for $i= 1, \ldots, K_u$. Below we below distinguish two cases.
\begin{enumerate}
\item If $u > 1$, then setting $N_{i_j}=\snr^{\frac{r_{i_j}}{2n_t}}$, $\kappa_{i_j}^2 = \snr^{1- \frac{r_{i_j}}{n_t}}$, and replacing $D_{\cal I}(N_{i_1}, \ldots, N_{i_u})$ by the lower bound in Theorem 8.3 gives
\begin{multline}
d_{\min}(H_{i_1}, \ldots, H_{i_u})
\dot\geq \max_{1 \leq t \leq K_u}\snr^{-\frac{1}{t} \sum_{i=K_u - t +1}^{K_u} \alpha_i } \snr^{\frac{1}{t} \sum_{j=1}^u \left(n_t - r_{i_j}\right)} \times
 \snr^{-\frac{(U-1)}{t} \sum_{j=1}^u r_{i_j} }  \snr^{-\frac{u n_t -t}{t}}. \label{eq:dminlb}
\end{multline}
\item If $u=1$ and ${\cal I}=\{i_1\}$, then similarly we have
\begin{multline*}
d_{\min}(H_{i_1})\  \dot\geq \ \max_{1 \leq t \leq K_u}\snr^{-\frac{1}{t} \sum_{i= K_u - t +1}^{K_u} \alpha_i } \times
 \snr^{\frac{1}{t} (n_t - r_{i_1})} \snr^{-\frac{n_t - t}{t}},
\end{multline*}
as $D_{\cal I}(N_{i_1})\ \dot\geq\ 1$ from Theorem 8.3.
\end{enumerate}

For simplicity, below we focus only on the case $u > 1$ as the analysis of the other case follows the same approach. Note that the norm $\norm{W}^2$ of the noise matrix is a $\chi^2$ random variable with $2n_r k$ degrees of freedom. Continuing from \eqref{eq:bdd}, we have
\begin{IEEEeqnarray*}{rCl}
&& \Pr \left\{ {\cal E}_{\cal I} \right\}\\
& \leq & \Pr \left\{ \norm{W} \geq \frac12 d_{\min}(H_{i_1}, \ldots, H_{i_u}) \right\}\\
& \doteq & \Pr \left\{ d_{\min}(H_{i_1}, \ldots, H_{i_u}) \ \dot\leq \ 1\right\}\\
& \stackrel{\text{(a)}}{\doteq} & \Pr \left\{\max_{1 \leq t \leq K_u} \left\{ \sum_{i=K_u - t +1}^{K_u} (1-\alpha_i ) - U \sum_{j=1}^u r_{i_j} \right\} \leq 0\right\}\\
&=& \Pr \left\{ \sum_{i=K_u - t +1}^{K_u} (1-\alpha_i )   \leq U \sum_{j=1}^u r_{i_j}, \ t=1,2,\ldots, K_u\right\}\\
&\stackrel{\text{(b)}}{\doteq}&
\snr^{-d^*_{u n_t, n_r} \left( U\sum_{i \in {\cal I}} r_i \right)},
\end{IEEEeqnarray*}
where (a) follows from the lower bound on $d_{\min}(H_{i_1}, \ldots, H_{i_u})$ in \eqref{eq:dminlb}, and to establish (b) a similar derivation can be found for example in \cite[p.3313]{Paw09}.
\end{proof}

\section{Conclusion}

While this paper concentrated on building and analyzing codes  for MIMO-MAC the mathematical methods we used and developed can hopefully
be applied to varied problems related to wireless communication.

Here we have collected  some of the main tools we used in this work to the following list.
\begin{itemize}
\item Pigeon hole principle  in subspace (Section \ref{pigeonsec}).
\item Use of valuation theory to  achieve linear independence ( Lemma \ref{determinanttivaluaatio}).
\item Use of Diophantine approximation to achieve Euclidean separation beyond linear independence (Theorem \ref{inertdecay3}).
\end{itemize}

Besides MIMO-MAC these methods might be useful in the study of single user MIMO and  interference channels.

\section*{Acknowledgement}
The authors would like to thank Tapani Matala-aho for his  help with the Diophantine approximation.

\end{document}